\documentclass[journal,draftcls,onecolumn,12pt,twoside]{IEEEtranTCOM}
\renewcommand{\baselinestretch}{1.62}  
\usepackage{moreverb}
\usepackage{epsfig}
\usepackage{amsmath,amssymb,amsthm,mathrsfs,amsfonts,dsfont}
\usepackage{adjustbox,lipsum}
\usepackage{algorithm,algorithmic}
\usepackage{amsfonts}
\usepackage{epsfig}
\usepackage{amssymb}
\usepackage{amsmath}
\usepackage{amsthm}
\usepackage{multirow}
\usepackage{setspace}
\usepackage{rotating}
\usepackage{graphicx}
\usepackage{tabularx}
\usepackage{array}
\usepackage{anyfontsize}
\usepackage{color,soul}
\usepackage{graphicx,dblfloatfix}
\usepackage{epstopdf}
\usepackage{blindtext}
\usepackage{amsmath}
\usepackage{amsthm,amssymb,amsmath,bm}
\usepackage{subfig}
\usepackage{amsfonts}
\usepackage{epsfig}
\usepackage{amssymb}
\usepackage{amsmath}
\usepackage{cite}
\usepackage{graphicx}
\usepackage{fancyhdr}
\usepackage[font={small}]{caption}
\usepackage{tabularx}
\usepackage{tcolorbox}
\usepackage{cite}
\usepackage{setspace}

\newtheorem{theorem}{Theorem}

\allowdisplaybreaks
\begin{document}
	
	\title{Average AoI in Multi-Source Systems with Source-Aware Packet Management 
	\thanks{ This research has been financially supported by the Infotech Oulu, the Academy of Finland (grant 323698), and Academy of Finland 6Genesis Flagship (grant 318927). M. Codreanu would like to acknowledge the support of the European Union's Horizon 2020 research and innovation programme under the Marie Sk\l{}odowska-Curie Grant Agreement No. 793402 (COMPRESS NETS). The work of M. Leinonen has also been financially supported in part by the Academy of Finland (grant 319485). 
		M. Moltafet would like to acknowledge the support of Finnish Foundation for Technology Promotion,  HPY Research Foundation, and Riitta ja Jorma J. Takanen Foundation.
}
	\thanks{Mohammad Moltafet and Markus Leinonen are with the Centre
		for Wireless Communications--Radio Technologies, University of Oulu,
		90014 Oulu, Finland (e-mail: mohammad.moltafet@oulu.fi; markus.leinonen@oulu.fi), and
		Marian Codreanu is with Department of Science and Technology, Link\"{o}ping University, Sweden (e-mail: marian.codreanu@liu.se)
	}
\thanks{	
	Preliminary results of this paper were presented in \cite{9155314}.
}
\author{
		Mohammad~Moltafet, Markus~Leinonen, and Marian~Codreanu
}
}
	\maketitle
	
	\vspace{-.5\baselineskip}
\begin{abstract}
We study the information freshness under three different source aware packet management policies  in a status update system consisting of two independent sources and one server. The packets of each source are generated according to the Poisson process  and the packets are served according to an exponentially distributed service time.  We derive the average age of information (AoI) of each source using the stochastic hybrid systems (SHS) technique for each packet management policy. In Policy 1, the queue can contain at most two waiting packets at the same time (in addition to the packet under service), one packet of source 1 and one packet of source 2. When the server is busy at an arrival of a packet, the possible packet of the same source waiting in the queue (hence, {source-aware}) is replaced by the arrived fresh packet. In Policy 2, the system (i.e., the waiting queue and the server)  can contain at most two packets, one from each source. When the server is busy at an arrival of a packet, the possible packet of the same source in the system is replaced by the fresh packet. Policy 3 is similar to Policy 2  but it does not permit  preemption in service, i.e., while a packet is under service all new arrivals from the same source are blocked and cleared.
Numerical results are provided to assess the fairness between sources and the sum average AoI of the proposed policies.

\end{abstract}

	\section{Introduction}\label{Introduction}
In many Internet of things applications and cyber-physical control systems, freshness of the status information at receivers is a critical factor. Recently, the age of information (AoI) was proposed as a destination-centric  metric   to measure the information freshness  in status update systems \cite{8469047,6195689,6310931}. A status update packet contains the measured value of a monitored process and a time stamp representing the time when the sample was generated. Due to wireless channel access, channel errors, and fading, etc., communicating a status update packet through the network experiences a random delay. If at a time instant $t$, the most recently received status update packet contains the time stamp $U(t)$, AoI is defined
as the random process $\Delta(t)=t-U(t)$. Thus, the AoI measures for each sensor the time elapsed since the last received status update packet was generated at the sensor. The average AoI is the most commonly used metric to evaluate the AoI \cite{8187436,8901143,moltafet2019power,6195689,6310931,5984917,7415972,8006592,7541764,6284003,9099557,8469047,8406928,8437591,8406966,8437907,9013935,9086190,8006504}.

The work \cite{6195689} is the seminal queueing theoretic work on the AoI in which the authors derived the average AoI for a single-source first-come first-served (FCFS) M/M/1 queueing model. 
In \cite{6875100}, the authors proposed peak AoI as an alternative  metric to evaluate the information freshness.
 The work \cite{6284003} was the first to investigate the average AoI in a multi-source setup. The authors of \cite{6284003} derived the average AoI for a multi-source FCFS M/M/1  queueing model.
The authors of \cite{7282742} considered a multi-source M/G/1 queueing system  and optimized the arrival rates of each source to minimize the peak AoI.  The authors of \cite{9099557} derived an exact expression for the average AoI for a multi-source FCFS M/M/1  queueing model and  an approximate expression for the average AoI for a multi-source FCFS M/G/1   queueing model having a general service time distribution.

The above works address the FCFS policy under the infinite queue size. However, it has been shown that the AoI can be significantly decreased   if there is a possibility to apply \emph{packet management} in the system (either in the queue or server) \cite{8469047,6310931,7415972,8006504,7541764,8406928,8437591,8406966,8437907,9013935}.  
To this end, the  average  AoI for a last-come first-served  (LCFS)  M/M/1  queueing model with preemption was analyzed in \cite{6310931}.
The average AoI for different packet management policies in a single-source M/M/1 queueing model were derived in \cite{7415972}. 
The authors of \cite{8006504}  derived a closed-form expression for the average
AoI of a single-source  M/G/1/1 preemptive queueing model (where the last entry in the Kendall notation shows the total capacity of the queueing system; 1 indicates that there is one packet under service whereas the queue holds zero packets).
The work \cite{7541764} considered a single-source LCFS queueing model where the packets arrive according to a Poisson process and the service time follows a gamma distribution. They derived the average AoI for two packet management policies, LCFS  with and  without preemption. The closed-form expressions for the average AoI and average peak AoI in a multi-source M/G/1/1 preemptive queueing model were derived in \cite{8406928}.

In \cite{8469047},  the authors  gave an in-depth introduction on a powerful technique, \textit{stochastic hybrid systems} (SHS), that can be used to evaluate the AoI in different  continuous-time queueing systems. They considered a multi-source queueing model in which the packets of different sources are generated according to the Poisson process and      
 served according to an exponentially distributed service time. The authors derived the average AoI for two packet management policies: 1) LCFS with preemption under service (LCFS-S), and 2) LCFS with preemption only
in waiting (LCFS-W). Under the LCFS-S policy, a new arriving packet preempts any packet that is currently
under service (regardless of the source index). Under the LCFS-W policy, a new arriving packet replaces
any older packet waiting in the queue (regardless of the source index); however, the new
packet has to wait for any packet under service to finish.

Since its establishment as an efficient tool for the AoI analysis \cite{8469047}, the SHS technique has recently been applied to derive the average AoI for various queueing models and packet management policies \cite{8437591,8406966,8437907,9013935,9086190,9103131}.
 {The authors of \cite{8437591} studied  a multi-source  M/M/1 queueing model in which sources have different priorities and proposed two packet management policies: 1) there is no waiting room
 and an update under service is preempted on arrival of an equal or
 higher priority update, and 2) there is a waiting room for at
 most one update and preemption is allowed in waiting but not in
 service.} In  \cite{8406966}, the author considered a single-source M/M/1 status update system in which the updates follow a route through a series of network nodes where each node is an LCFS queue that supports preemption in service. 
In \cite{8437907}, the author considered a single-source   LCFS queueing model with multiple servers with preemption in service.
The authors of  \cite{9013935} considered a multi-source LCFS queueing model with multiple servers that employ preemption in service. In \cite{9086190}, the authors derived the average AoI for  a multi-source FCFS M/M/1 queueing model with an infinite queue size. In \cite{9103131}, the authors studied moments and the moment generating function of the AoI.

\subsection{Contributions}
In this paper, we consider a status update system in which two independent sources generate packets according to the Poisson process  and the packets are served according to an exponentially distributed service time. To emphasize the importance of minimizing the average AoI of each individual source and  enhance the fairness between different sources in the system, we propose three different \emph{source-aware} packet management policies.

In Policy 1, the queue can contain at most two waiting packets at the same time (in addition to the packet under service), one packet of source 1 and one packet of source 2.  When the server is busy and a new packet  arrives, the possible packet of the same source waiting in the queue (not being served) is replaced by the fresh packet.
In Policy 2, the system (i.e., the waiting queue and the server)  can contain at most two packets, one from each source. When the server is busy at an arrival of a packet, the possible packet of the same source  either
waiting in the queue or being served  is replaced by the fresh packet.
 Policy 3 is similar to Policy 2  but it does not permit  preemption in service, i.e., while a packet is under service all new arrivals from the same source are blocked and cleared.

We derive the average AoI for each source under the proposed packet management policies using the SHS technique. By numerical experiments, we investigate the effectiveness  of the  proposed packet management policies in terms of the sum average AoI and fairness between different sources. The results show that our proposed policies  provide better fairness than that of the existing policies. In addition, Policy 2 outperforms the existing policies in terms of the sum average AoI. To the best of our knowledge, the proposed policies have not been considered and analyzed in the AoI context earlier.

\subsection{Related Works}
The most related works to our paper are \cite{8469047,8437591,9013935}. Each of these works considers a multi-source queueing model with such a packet management policy where the different sources can preempt the packets of each other in the system (\emph{source-agnostic} preemption). In addition, in \cite{8437591}, an arriving packet is discarded if the server is currently serving a packet of another source having higher priority (service priority). These packet management policies result in queueing systems where the server mostly serves either the packets of a source with a high packet arrival rate  \cite{8469047,9013935}, or the packets of a source with high service priority \cite{8437591}. In this regard,  these policies are not appropriate for the applications in which besides the sum average AoI,  the average AoI of each individual  source is important. Note that differently from \cite{8469047,8437591,9013935}, all our policies employ source-aware preemption in the system in the sense that an arriving packet can preempt \emph{only} a packet with the same source index in the system. This promotes fairness, as shown by our numerical experiments.

\subsection{Organization}
The paper is organized as follows. The system model and problem definition are  presented in Section \ref{System Model}.
The basics of the SHS technique are presented in Section \ref{Introduction to the SHS Technique}. The average AoI  for each source under different packet management policies is derived in Section \ref{AoI Analysis Using the SHS Technique}.
Numerical results are presented in Section \ref{Numerical Results}.
Finally, concluding remarks are expressed in Section \ref{Conclusions}.

 \section{System Model and Summary of the Main Results}\label{System Model}
We  consider a status update system consisting of two independent sources\footnote{We consider two sources for simplicity of presentation; the same methodology as used in this paper can be applied for more than two sources. However, the complexity of the calculations increases exponentially with the number of sources.}, one server, and one  sink, as depicted in Fig. \ref{AoIs}.
Each source observes a random process at random time instants. The  sink is interested in timely information about the status of these random processes. Status updates are transmitted as packets, containing the measured value of the monitored process and a time stamp representing the time when the sample was generated. We assume that  the packets  of sources  1 and 2   are generated according to the Poisson process with rates  $\lambda_1$ and $\lambda_2$, respectively, and      
the packets are served according to an exponentially distributed service time with mean ${1}/{\mu}$. Let  $\rho_1={\lambda_1}/{\mu}$ and $\rho_2={\lambda_2}/{\mu}$ be the  load of  source 1 and 2, respectively. Since packets of the sources are generated  according to the Poisson process and the sources are independent, the  packet generation in the system follows the Poisson process with rate
$\lambda=\lambda_1+\lambda_2$. The overall load in the system is
$\rho=\rho_1+\rho_2={\lambda}/{\mu}$.

In the next subsections, we first explain each packet management policy, and then, give a formal definition of AoI. 
	\begin{figure}[t]
	\centering
	\subfloat[Policy 1: The queue can contain at most two waiting packets at the same time (in addition to the packet under service), one packet of source 1 and one packet of source 2;  when the server is busy and a new packet  arrives, the possible packet of the same source waiting in the queue (not being served) is replaced by the fresh packet.]{\includegraphics[width=0.79\linewidth]{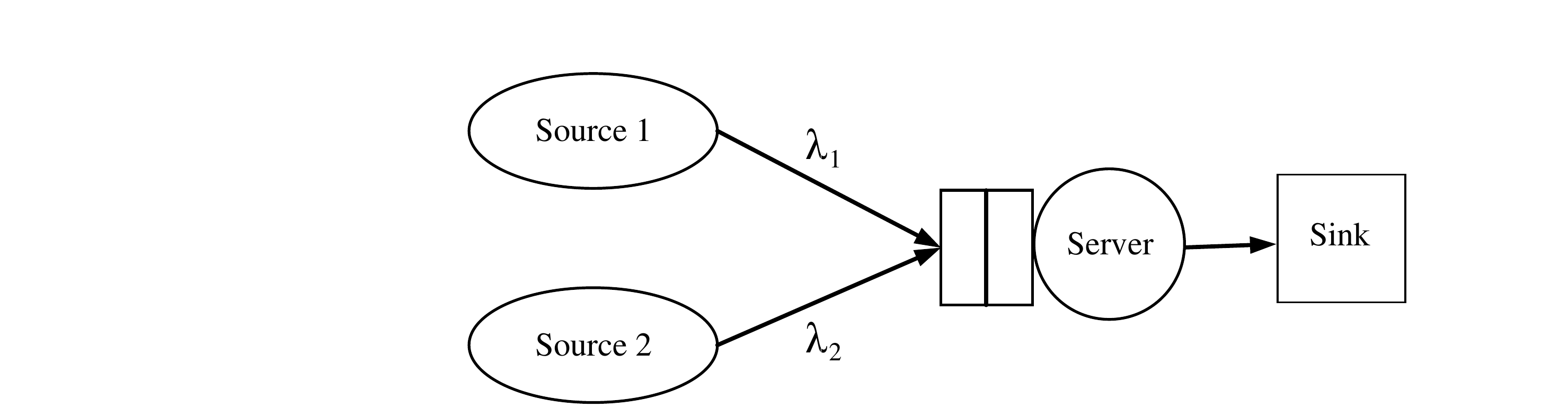}}\\
	\subfloat[ Policies 2 and 3: The system (i.e., the waiting queue and the server)  can contain at most two packets, one from each source.	In Policy 2, when the server is busy and a new packet arrives, the possible packet of the same source  either
	waiting in the queue or being served  is replaced by the fresh packet. Policy 3 is similar to Policy 2  but it does not permit  preemption in service.]{\includegraphics[width=0.76\linewidth]{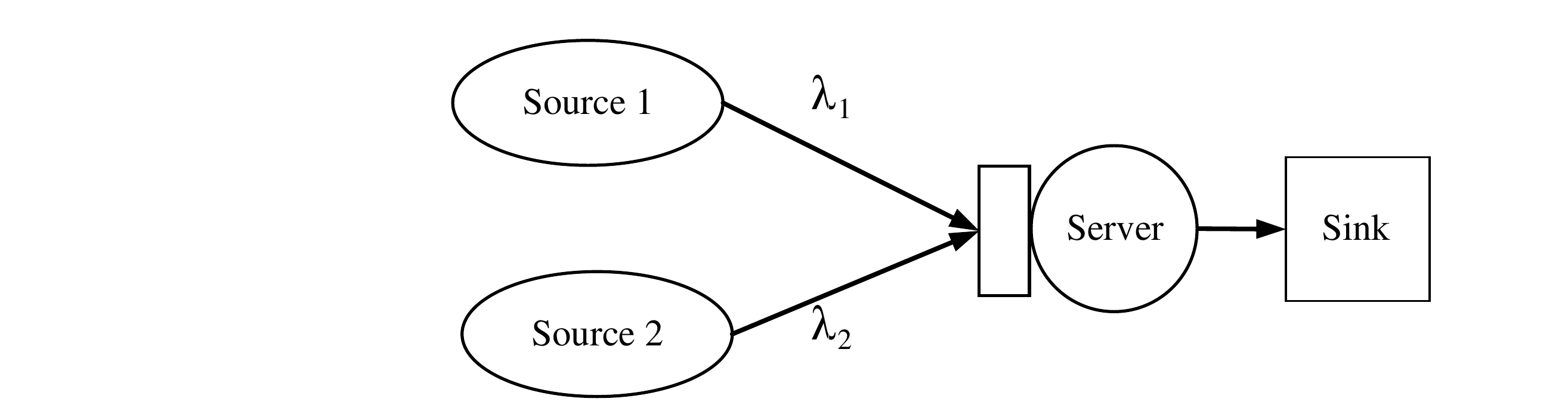}}\vspace{-4mm}
	\caption{
The packet management policies. 
	}  
		\vspace{-12mm}			
	\label{AoIs}
\end{figure}

\subsection{Packet Management Policies}\label{II-A}

The structure of the queuing system for all considered policies is illustrated in Fig.~\ref{AoIs}. In all policies, when the system is empty, any arriving packet immediately enters the server. However, the policies differ in how they handle the arriving packets when the server is busy.

In Policy 1  (see Fig.~\ref{AoIs}(a)), when the server is busy and a new packet arrives, the possible packet of the same source waiting in the queue (not being served) is replaced by the fresh packet.

In Policy 2 (see Fig.~\ref{AoIs}(b)), when the server is busy and a new packet arrives, the possible packet of the same source  either waiting in the queue or being served (called \textit{self preemption}) is replaced by the fresh packet.

Policy 3 is similar to Policy 2  but it does not permit  preemption in service (see Fig.~\ref{AoIs}(b)). While a packet is under service all new arrivals from the same source are blocked and cleared. However, the packet waiting in the queue is replaced upon the arrival of a newer one from the same source. It is also interesting to remark that this policy is also similar to Policy 1 but it has a one unit shorter waiting queue.

	\subsection{AoI Definition}
	For each source, the AoI  at the destination is defined as the time elapsed since the last  successfully received packet was generated. Formal definition of the AoI is given next.
	
Let $t_{c,i}$ denote the time instant at which the $i$th status   update packet of source $c$ was generated, and $t'_{c,i}$ denote the time instant at which this packet  arrives at the sink. At a time instant $\tau$,  the  index of the most recently received packet of source $c$ is given by
$
N_c(\tau)=\max\{i'|t'_{c,i'}\le \tau\},
$
and the time stamp of the most recently received packet of source $c$ is
$
U_c(\tau)=t_{c,N_c(\tau)}.
$
The AoI of source $c$ at the destination is defined as the random process
$
\Delta_{c}(t)=t-U_c(t).
$
Let $(0,\tau)$ denote   an observation interval. Accordingly,  the time average AoI of  the source $c$ at the sink, denoted as  $\Delta_{\tau,c}$, is  defined as
$$
\displaystyle\Delta_{\tau,c}=\dfrac{1}{\tau}\int_{0}^{\tau}\Delta_{c}(t)\mathrm{d}t.
$$
The  average AoI of source $c$, denoted by $\Delta_{c}$, is defined as
\begin{equation}\label{AoIeq}
\Delta_{c}=\lim_{\tau\to\infty}\Delta_{\tau,c}.
\end{equation}

\subsection{Summary of the Main Results}
In this paper, we derive the average AoI for each source under the Policy 1, Policy 2, and Policy 3 using the SHS technique.  The derived results are summarized by the following three theorems.

\begin{theorem}\label{theo1}
	 The average AoI of source 1 under Policy 1 is given as
\begin{align}\nonumber
\Delta_{1}=\dfrac{\sum_{k=0}^7 \rho_1^k\eta_k}{\mu \rho_1\left(1+\rho_1\right)^2\sum_{j=0}^4 \rho_1^j\xi_j},
\end{align}
where
\begin{align} \nonumber
&\eta_0=\rho_2^4+\!2\rho_2^3+3\rho_2^2+2\rho_2+1,~~~~~~~~~~~
\eta_1=7\rho_2^4+15\rho_2^3+21\rho_2^2+14\rho_2+6,\\\nonumber&
\eta_2=17\rho_2^4+46\rho_2^3+64\rho_2^2+42\rho_2+16,\,\,\,\,\,
\eta_3=15\rho_2^4+73\rho_2^3+118\rho_2^2+78\rho_2+26,\\\nonumber&
\eta_4=5\rho_2^4+52\rho_2^3+124\rho_2^2+102\rho_2+30,\,\,\,\,
\eta_5=15\rho_2^3+66\rho_2^2+79\rho_2+24,\\\nonumber&
\eta_6=15\rho_2^2+31\rho_2+11,~~~~~~~~~~~~~~~~~~~~~\eta_7=5\rho_2+2,\\\nonumber&
\xi_0=\rho_2^4+2\rho_2^3+3\rho_2^2+2\rho_2+1,~~~~~~~~~~~~~
\xi_1=2\rho_2^4+6\rho_2^3+9\rho_2^2+7\rho_2+3,\,\,\\\nonumber&
\xi_2=6\rho_2^3+12\rho_2^2+10\rho_2+4,~~~~~~~~~~~~~~~~
\xi_3=6\rho_2^2+8\rho_2+3,\,\,\,\\\nonumber&
\xi_4=2\rho_2+1.
\end{align}
\end{theorem}
\begin{proof}
The proof of Theorem \ref{theo1} appears in  Section \ref{FCFS Prioritized Packet Management Policy} of this paper.
\end{proof}

\begin{theorem}\label{theo2}
	The average AoI of source 1 under Policy 2 is given as
	\begin{align}\nonumber
	\Delta_{1}=\dfrac{{(\rho_2+1)}^2+\sum_{k=1}^5 \rho_1^k\tilde\eta_k}{\mu \rho_1\left(1+\rho_1\right)^2\big(\rho_1^2(2\rho_2+1)+(\rho_2+1)^2(2\rho_1+1)\big)},
	\end{align}
	where
	\begin{align} \nonumber
	&\tilde\eta_1=6\rho_2^2+11\rho_2+5,~~~~~~~~~~~~
	\tilde\eta_2=13\rho_2^2+24\rho_2+10,\\\nonumber&
	\tilde\eta_3=10\rho_2^2+27\rho_2+10,~~~~~~~~~
	\tilde\eta_4=3\rho_2^2+14\rho_2+5,\\\nonumber&
	\tilde\eta_5=3\rho_2+1.
	\end{align}
\end{theorem}
\begin{proof}
	The proof of Theorem \ref{theo2} appears in  Section \ref{Average AoI Under Policy 2} of this paper.
\end{proof}

\begin{theorem}\label{theo3}
	The average AoI of source 1 under Policy 3 is given as
	\begin{align}\nonumber
	\Delta_{1}=\dfrac{{(\rho_2+1)}^3+\sum_{k=1}^4 \rho_1^k\hat\eta_k}{\mu \rho_1\left(1+\rho_1\right)\left(1+\rho_2\right)\big(\rho_1^2(2\rho_2+1)+(\rho_2+1)^2(2\rho_1+1)\big)},
	\end{align}
	where
	\begin{align} \nonumber
	&\hat\eta_1=5\rho_2^3+14\rho_2^2+13\rho_2+4,~~~~~~~~~~~~
	\hat\eta_2=10\rho_2^3+28\rho_2^2+25\rho_2+7,\\\nonumber&
	\hat\eta_3=5\rho_2^3+22\rho_2^2+23\rho_2+6,~~~~~~~~~~~~
	\hat\eta_4=5\rho_2^2+8\rho_2+2.\\\nonumber&
	\end{align}
\end{theorem}
\begin{proof}
	The proof of Theorem \ref{theo3} appears in  Section \ref{Average AoI Under Policy 3} of this paper.
\end{proof}

\section{ A Brief Introduction to the SHS Technique}\label{Introduction to the SHS Technique}	
In the following, we briefly present the main idea  behind the SHS technique which is the key tool for our AoI analysis in Section \ref{AoI Analysis Using the SHS Technique}. We refer the readers to \cite{8469047} for more details.

The SHS technique models a queueing system  through the states $(q(t), \bold{x}(t))$, where  ${q(t)\in \mathcal{Q}=\{0,1,\ldots,m\}}$ is a continuous-time finite-state Markov
chain that  
describes the occupancy of the system and ${\bold{x}(t)=[x_0(t)~x_1(t)\cdots x_n(t)]\in \mathbb{R}^{1\times(n+1)}}$ is a continuous
process that describes the evolution of age-related processes at the sink. Following the approach in  \cite{8469047}, we label the source of interest as source 1 and  employ the continuous process $ \bold{x}(t) $ to track the age of source 1 status updates at the sink.

The Markov chain $q(t)$ can be presented as a graph $(\mathcal{Q},\mathcal{L})$ where 
each discrete state $q(t)\in \mathcal{Q}$ is a node of the chain and a (directed) link $ l\in\mathcal{L} $  from node $ q_l $ to node $q'_{l}$ indicates a transition from state $ {q_l \in \mathcal{Q}}$ to state ${q'_{l}\in \mathcal{Q}}$.

A transition occurs  when a packet arrives or departs in the system. Since the time elapsed between departures and arrivals is exponentially distributed according to the M/M/1 queueing model, transition $l\in\mathcal{L}$ from state $ q_l $ to state $q'_{l}$ occurs with the  exponential rate $\lambda^{(l)}\delta_{q_l,q(t)}$\footnote{In our system model, $ \lambda^{(l)} $ can represent three quantities: arrival rate of source 1 ($ \lambda_1 $), arrival rate of source 2 ($ \lambda_2 $), and the service rate ($ \mu $). },
 where the Kronecker delta function $\delta_{q_l,q(t)}$ ensures that the transition $ l $ occurs only when the discrete
state $ q(t) $ is equal to $ q_l $.  When a transition $l$ occurs, the  discrete state $ q_l $ changes   to state $q'_{l}$, and the continuous state $\bold{x}$ is reset to $\bold{x}'$ according to a binary transition reset map matrix ${\bold{A}_l}\in\mathbb{B}^{(n+1)\times(n+1)}$ as ${\bold{x}'=\bold{x}\bold{A}_l}$. In addition, at each state ${q(t)=q\in \mathcal{Q}}$, the continuous state $\bold{x}$ evolves as a piece-wise linear function through the 
differential equation ${\dot{\bold{x}}(t)\triangleq\dfrac{\partial\bold{x}(t)}{\partial  t}=\bold{b}_q}$, where
$\bold{b}_q=[b_{q,0}~b_{q,1}\cdots b_{q,n}]\in\mathbb{B}^{1\times(n+1)}$ is a binary vector with  elements $b_{q,j}\in \{0,1\}, \forall j\in\{0,\ldots,n\},q \in\mathcal{Q}$.
If the age process $x_j(t)$ increases at a unit rate, we have ${b}_{q,j}=1$; otherwise,   ${b}_{q,j}=0$.  

Note that unlike in a typical continuous-time Markov
chain, a transition
 from a  state to itself  (i.e., a self-transition) is possible in  $q(t)\in \mathcal{Q}$. In the case of a self-transition, a reset of the continuous state $\bold{x}$ takes place, but the discrete state remains the same. In addition, for a given pair of states $s,s'\in \mathcal{Q}$, there  may be multiple
transitions $ l $ and $ l' $ so that the discrete state changes from $ s $ to
$ s' $ but the transition reset maps $\bold{A}_l $ and $ \bold{A}_{l'}$ are different (for more details, see  \cite[Section III]{8469047}). 

To calculate the average AoI using the SHS technique, the state probabilities of the Markov chain and the correlation vector between the discrete state $q(t)$ and the continuous state $\bold{x}(t)$ need to be calculated. Let $\pi_q(t)$ denote the probability of being in state $q$ of the Markov chain and $\bold{v}_q(t)=[{v}_{q0}(t)\cdots{v}_{qn}(t)]\in\mathbb{R}^{1\times(n+1)}$ denote the correlation vector between the discrete state $q(t)$ and the continuous state $\bold{x}(t)$. Accordingly, we have 
\begin{equation}
\pi_q(t)=\mathrm{Pr}(q(t)=q)=\mathbb{E}[\delta_{q,q(t)}], \,\,\,\forall q\in\mathcal{Q},
\end{equation}
\begin{equation}
\bold{v}_q(t)=[{v}_{q0}(t)\cdots{v}_{qn}(t)]=\mathbb{E}[\bold{x}(t)\delta_{q,q(t)}],\,\,\,\forall q\in\mathcal{Q}.
\end{equation}

Let $\mathcal{L}'_q$ denote the set of incoming transitions and  $\mathcal{L}_q$ denote the set of outgoing transitions for state $q$, defined as 
\begin{align}\nonumber
&\mathcal{L}'_q=\{l\in\mathcal{L}:q'_{l}=q\},\,\,\,\forall q\in\mathcal{Q},\\\nonumber&\mathcal{L}_q=\{l\in\mathcal{L}:q_{l}=q\},\,\,\,\forall q\in\mathcal{Q}.
\end{align}
Following the ergodicity assumption of the Markov chain $q(t)$ in the AoI analysis \cite{8469047,9103131,9007478}, the state
probability vector $\boldsymbol{\pi}(t)=[\pi_0(t) \cdots \pi_m(t)]$ converges uniquely 
to the  stationary vector $\bar{\boldsymbol{\pi}}=[\bar{\pi}_0 \cdots \bar{\pi}_m]$ satisfying \cite{8469047}
\begin{align}\label{eqrt01}
&\bar{{\pi}}_q\textstyle\sum_{l\in\mathcal{L}_q}\lambda^{(l)}=\textstyle\sum_{l\in\mathcal{L}'_q}\lambda^{(l)}\bar{{\pi}}_{q_l}, \,\,\,\forall q\in\mathcal{Q},\\&\label{erwq}
\textstyle\sum_{q\in\mathcal{Q}}\bar{{\pi}}_q=1.
\end{align}
Further,  it has been shown in \cite[Theorem 4]{8469047} that under the ergodicity assumption of the Markov chain $q(t)$ with stationary distribution $\bar{\boldsymbol{\pi}}\succ0$, the existence of a nonnegative solution ${\bar{\bold{v}}_q=[\bar{v}_{q0} \cdots \bar{v}_{qn} ], \forall q\in \mathcal{Q}},$ for the following system of linear equations
	\begin{align}\label{asleq}	
	\bar{\bold{v}}_q\textstyle\sum_{l\in\mathcal{L}_q}\lambda^{(l)}=\bold{b}_q\bar{{\pi}}_q+\textstyle\sum_{l\in\mathcal{L}'_q}\lambda^{(l)}\bar{\bold{v}}_{q_l}\bold{A}_l, \,\,\,\forall q\in\mathcal{Q},
	\end{align}
	implies that the correlation vector $\bold{v}_q(t)$ converges to  ${\bar{\bold{v}}_q=[\bar{v}_{q0} \cdots \bar{v}_{qn} ], \forall q\in \mathcal{Q}}$  as $t\rightarrow\infty$. 
	Finally, the average AoI of source 1 is calculated by \cite[Theorem 4]{8469047}
	\begin{align}\label{AOIANAL}	
	\Delta_1=\textstyle\sum_{q\in\mathcal{Q}}\bar{v}_{q0}.
	\end{align}

As \eqref{AOIANAL} implies, the main challenge in calculating the average AoI of a source using the SHS technique reduces to deriving the first elements of each correlation vector $\bar{\mathbf{v}}_{q}$, i.e., 
$\bar{v}_{q0}$, ${\forall{q}\in\mathcal{Q}}$. Note that these quantities are, in general, different for each particular  queueing model.

\section{ Average AoI Analysis Using the SHS Technique}\label{AoI Analysis Using the SHS Technique}	
 In this section, we use  the SHS technique to calculate the average AoI in \eqref{AoIeq} of each source under the considered packet management policies described in Section \ref{II-A}. Recall from \eqref{AOIANAL} that the characterization of the average AoI in each of our queueing setup is accomplished by deriving the quantities $\bar{v}_{q0}$, ${\forall{q}\in\mathcal{Q}}$. The next three sections are devoted to elaborate derivations of these quantities.
\subsection{Average AoI under Policy 1}\label{FCFS Prioritized Packet Management Policy}

In Policy 1, the state space of the Markov chain is $\mathcal{Q}=\{0,1,\ldots,5\}$, with each state presented in Table \ref{table-1}. For example, ${q=0}$ indicates that the server is idle which is shown by I; ${q=1}$ indicates that a packet is under service, i.e., the queue is empty and the server is busy which is shown by B; and  ${q=5}$ indicates that server is busy, the first packet in the queue (i.e., the packet that is at the head of the queue as depicted in Fig. \ref{AoIs}(a)) is a source 2  packet, and the second packet in the queue is a source 1  packet.

The continuous process is ${\bold{x}(t)=[x_0(t)~x_1(t)~x_2(t)~x_3(t)]}$, where $x_0(t)$ is the current AoI of source 1 at time instant $t$, $\Delta_1(t)$;  $ x_1(t)  $ encodes  what $\Delta_1(t)$ would become if the packet that is under service is delivered to the sink at time instant $t$; $ x_2(t)  $ encodes  what $\Delta_1(t)$ would become if the first packet in the queue is delivered to the sink at time instant $t$; $ x_3(t)  $ encodes  what $\Delta_1(t)$ would become if the second packet in the queue is delivered to the sink at time instant $t$.

Recall that our goal is to  find $\bar{v}_{q0}, \forall q\in\mathcal{Q} $, to calculate the average AoI of source 1 in \eqref{AOIANAL}. To this end, we need to solve the system of linear equations \eqref{asleq} with variables $\bar{\bold{v}}_q, \forall{q}\in\mathcal{Q}$. To form the system of linear equations \eqref{asleq} for each state $\forall{q}\in\mathcal{Q}$, we need to determine  $\bold{b}_q$, $\bar{{\pi}}_q$, and  $\bar{\bold{v}}_{q_l}\bold{A}_l$ for each incoming transition ${l\in\mathcal{L}'_q}$. Next, we derive these for Policy 1.

\subsubsection{Determining the value of $\bar{\bold{v}}_{q_l}\bold{A}_l$ for incoming transitions for each state $q\in\mathcal{Q}$}
The Markov chain for the discrete state $q(t)$ with the incoming and outgoing transitions for each state $q\in\mathcal{Q}$ is shown in Fig. \ref{Chain1}. The transitions between the discrete states ${{q_l \rightarrow q'_l}, \,\,\forall l\in \mathcal{L}}$, and their effects on the continuous state $\bold{x}(t)$ 
are summarized in Table \ref{table-2}. In the following, we explain the transitions presented in Table \ref{table-2}:

\begin{figure}
	\centering
	\includegraphics[scale=.7]{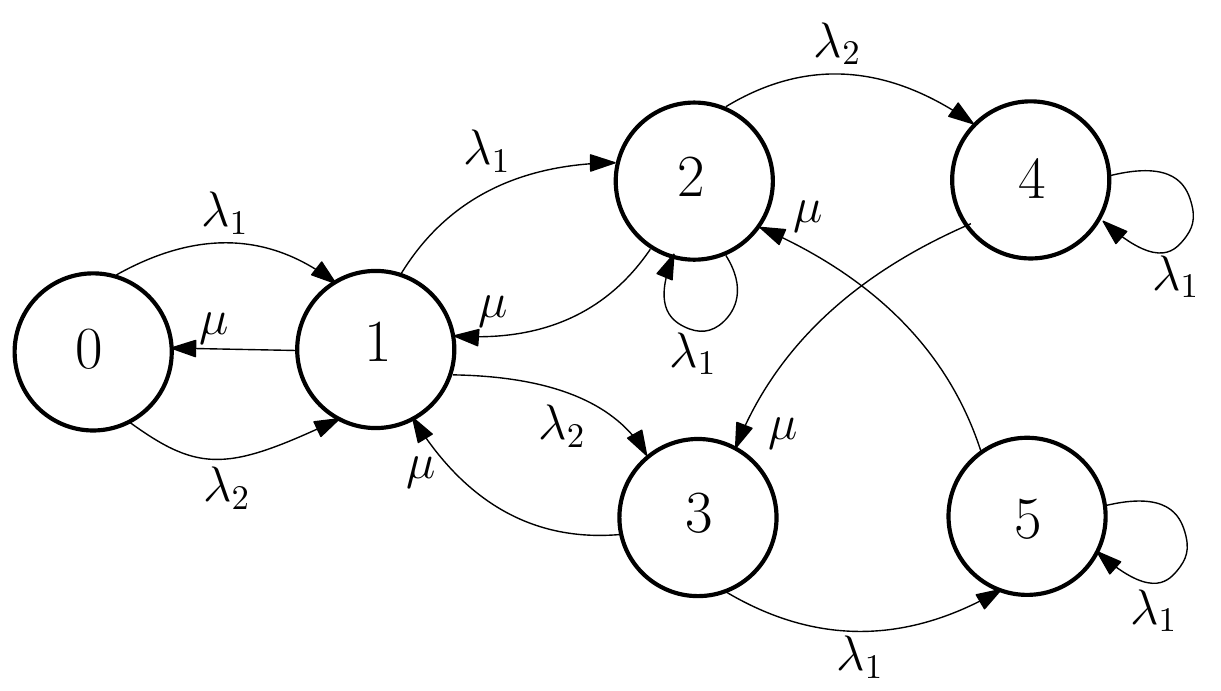}
	\caption{The SHS Markov chain for Policy 1.}
		\vspace{-5mm}
	\label{Chain1}
\end{figure}

\begin{table}
	\centering\small
	\caption{SHS Markov chain states for Policy 1}
	\label{table-1}
	\begin{tabular}{ |c|m{3.5cm}|m{3.5cm}|c|}
		\hline
		State  & Source index of the second packet in the queue  &Source index of the first packet in the queue & Server \\			
		\hline			
		0&-& -&I	\\		
		\hline
		1&-&-&B\\
		\hline	
		2&-&1&B\\
		\hline	
		3&-&2&B\\
		\hline	
		4&2&1&B\\
		\hline	
		5&1&2&B\\
		\hline				
	\end{tabular}
\end{table}

\begin{table}\small
	\centering
	\caption{Table of transitions for the Markov chain of Policy 1 in Fig. \ref{Chain1}}
	\label{table-2}
	\begin{tabular}{ |l|l|c|c|c|c|}
		\hline
		\textit{l}  & $q_l \rightarrow q'_l $&$\lambda^{(l)}$& $\bold{A}_l$&$\bold{x}\bold{A}_l$&$\bold{v}_{q_l}\bold{A}_l$ \\			
		\hline			
		1&$0 \rightarrow  1$& $\lambda_1$&$\left[x_0~ 0~ x_2 ~x_3\right]$&$\tiny\begin{bmatrix}
			1 & 0 & 0&0\\
			0 & 0 & 0&0\\
			0 & 0 & 1&0\\
			0 & 0 & 0&1
		\end{bmatrix}$&$\left[v_{00}~ 0~ v_{02} ~v_{03}\right]$	\\		
		\hline
	2&$0 \rightarrow  1$&$\lambda_2$&$\left[x_0~ x_0~ x_2~x_3\right]$&$\tiny\begin{bmatrix}
	1 & 1 & 0&0\\
	0 & 0 & 0&0\\
	0 & 0 & 1&0\\
	0 & 0 & 0&1
	\end{bmatrix}$&$\left[v_{00}~ v_{00}~ v_{02} ~v_{03}\right]$\\
		\hline	
			3&$1 \rightarrow  0$&$\mu$&$\left[x_1~ x_1~ x_2 ~x_3\right]$&$\tiny\begin{bmatrix}
		0 & 0 & 0&0\\
		1 & 1 & 0&0\\
		0 & 0 & 1&0\\
		0 & 0 & 0&1
		\end{bmatrix}$&$\left[v_{11}~ v_{11}~ v_{12} ~v_{13}\right]$\\		
		\hline	
		4&$1 \rightarrow  2$&$\lambda_1$&$\left[x_0~ x_1~ 0 ~x_3\right]$&$\tiny\begin{bmatrix}
		1 & 0 & 0&0\\
		0 & 1 & 0&0\\
		0 & 0 & 0&0\\
		0 & 0 & 0&1
		\end{bmatrix}$&$\left[v_{10}~ v_{11}~ 0 ~v_{13}\right]$\\
		\hline	
		5&$1 \rightarrow 3 $&$\lambda_2$&$\left[x_0~ x_1~ x_1 ~x_3\right]$&$\tiny\begin{bmatrix}
    	1 & 0 & 0&0\\
    	0 & 1 & 1&0\\
    	0 & 0 & 0&0\\
    	0 & 0 & 0&1
    	\end{bmatrix}$&$\left[v_{10}~ v_{11}~ v_{11} ~v_{13}\right]$\\
		\hline			
		6&$2 \rightarrow  1$&$\mu$&$\left[x_1~ x_2~ x_2 ~x_3\right]$&$\tiny\begin{bmatrix}
		0 & 0 & 0&0\\
		1 & 0 & 0&0\\
		0 & 1 & 1&0\\
		0 & 0 & 0&1
		\end{bmatrix}$&$\left[v_{21}~ v_{22}~ v_{22} ~v_{23}\right]$\\	
		\hline	
		7&$3 \rightarrow  1$&$\mu$&$\left[x_1~ x_1~ x_2 ~x_3\right]$&$\tiny\begin{bmatrix}
		0 & 0 & 0&0\\
		1 & 1 & 0&0\\
		0 & 0 & 1&0\\
		0 & 0 & 0&1
		\end{bmatrix}$&$\left[v_{31}~ v_{31}~ v_{32} ~v_{33}\right]$\\
		\hline	
		8&$2 \rightarrow  2$&$\lambda_1$&$\left[x_0~ x_1~ 0 ~x_3\right]$&$\tiny\begin{bmatrix}
		1 & 0 & 0&0\\
		0 & 1 & 0&0\\
		0 & 0 & 0&0\\
		0 & 0 & 0&1
		\end{bmatrix}$&$\left[v_{20}~ v_{21}~ 0 ~v_{23}\right]$\\
		\hline	
		9&$2 \rightarrow 4 $&$\lambda_2$&$\left[x_0~ x_1~ x_2 ~x_2\right]$&$\tiny\begin{bmatrix}
		1 & 0 & 0&0\\
		0 & 1 & 0&0\\
		0 & 0 & 1&0\\
		0 & 0 & 0&1
		\end{bmatrix}$&$\left[v_{20}~ v_{21}~ v_{22} ~v_{22}\right]$\\
		\hline	
		10&$3 \rightarrow  5$&$\lambda_1$&$\left[x_0~ x_1~ x_1 ~0\right]$&$\tiny\begin{bmatrix}
		1 & 0 & 0&0\\
		0 & 1 & 1&0\\
		0 & 0 & 0&0\\
		0 & 0 & 0&0
		\end{bmatrix}$&$\left[v_{30}~ v_{31}~ v_{31} ~0\right]$\\
		\hline	
		11&$4\rightarrow  4$&$\lambda_1$&$\left[x_0~ x_1~ 0 ~0\right]$&$\tiny\begin{bmatrix}
		1 & 0 & 0&0\\
		0 & 1 & 0&0\\
		0 & 0 & 0&0\\
		0 & 0 & 0&0
		\end{bmatrix}$&$\left[v_{40}~ v_{41}~ 0 ~0\right]$\\
		\hline	
		12&$5 \rightarrow  5$&$\lambda_1$&$\left[x_0~ x_1~ x_1 ~0\right]$&$\tiny\begin{bmatrix}
		1 & 0 & 0&0\\
		0 & 1 & 1&0\\
		0 & 0 & 0&0\\
		0 & 0 & 0&0
		\end{bmatrix}$&$\left[v_{50}~ v_{51}~ v_{51} ~0\right]$\\
		\hline	
		13&$4 \rightarrow  3$&$\mu$&$\left[x_1~ x_2~ x_2 ~x_3\right]$&$\tiny\begin{bmatrix}
		0 & 0 & 0&0\\
		1 & 0 & 0&0\\
		0 & 1 & 1&0\\
		0 & 0 & 0&1
		\end{bmatrix}$&$\left[v_{41}~ v_{42}~ v_{42} ~v_{43}\right]$\\
		\hline	
		14&$5 \rightarrow  2$&$\mu$&$\left[x_1~ x_1~ x_3 ~x_3\right]$&$\tiny\begin{bmatrix}
		0 & 0 & 0&0\\
		1 & 1 & 0&0\\
		0 & 0 & 0&0\\
		0 & 0 & 1&1
		\end{bmatrix}$&$\left[v_{51}~ v_{51}~ v_{53} ~v_{53}\right]$\\
		\hline			
	\end{tabular}
\end{table}

\begin{itemize}
	\item \textit{l=1}: A source 1 packet  arrives at an empty system. With this arrival/transition the AoI of source 1 does not change, i.e., $x'_0=x_0$.
	This is because the arrival of source 1 packet does not yield an age reduction until it is delivered to the sink. Since the arriving source 1 packet is fresh and its age is zero, we have $x'_1=0$. Since with this arrival the queue is still empty, $x_2$ and $x_3$ become irrelevant to the AoI of source 1, and thus, $x'_2=x_2$ and $x'_3=x_3$. Note that if the system moves into a new state where $ x_j $ is irrelevant, we set $ x'_j=x_j, j\in\{1,2,3\} $. 
	 An interpretation of this assignment is that  $ x_j $ has not changed in the transition to  the new state.
	 Finally, we have 
	\begin{align}\label{l1}
	\bold{x}'=[x_0~x_1~x_2~x_3]\bold{A}_1=[x_0~0~x_2~x_3].
	\end{align} 
	According to \eqref{l1}, it can be shown that the binary matrix $\bold{A}_1$ is given by
	\begin{align}\label{l11}
	\bold{A}_1=\begin{bmatrix}
	1 & 0 & 0&0\\
	0 & 0 & 0&0\\
	0 & 0 & 1&0\\
	0 & 0 & 0&1
	\end{bmatrix}.
	\end{align} 
	Then, by using \eqref{l11}, $\bold{v}_{0}\bold{A}_1$ is calculated as
	\begin{align}\label{l12}
	\bold{v}_{0}\bold{A}_1\!\!=\![v_{00}~ v_{01}~ v_{02} ~v_{03}]\bold{A}_1
	=\!\!\left[v_{00}~ 0~ v_{02} ~v_{03}\right].
	\end{align} 
	It can be seen from \eqref{l1}-\eqref{l12} that when we have $\bold{x}'$ for a transition $l\in\mathcal{L}$, it is easy to calculate $ \bold{v}_{q_l}\bold{A}_l $.
	Thus, for the rest of the transitions, we just explain the calculation of $\bold{x}'$ and present the final expressions of $\bold{A}_l$ and $ \bold{v}_{q_l}\bold{A}_l $.
	\item \textit{l=2}: A  source 2 packet arrives at an empty system.  We have $x'_0=x_0$, because this arrival does not change the AoI at the sink.  Since the arriving packet is a source 2 packet, its delivery does not change the AoI of source 1, thus we have
	$x'_1=x_0$. Moreover, since the queue is empty, $x_2$ and $x_3$ become irrelevant, and we have  $x'_2=x_2$ and $x'_3=x_3$.

	\item \textit{l=3}: A packet is under service and it completes service and is delivered
	to the sink. With this transition, 
	the AoI at the sink is reset to the age of the  packet that just completed service, and thus, $x'_0=x_1$. Since the 
	system enters state $q=0$,
	we have  ${x'_1=x_1}$, ${x'_2=x_2}$, and ${x'_3=x_3}$.

	\item \textit{l=4}: A  packet is under service and a source 1 packet arrives. In this transition, we have $x'_0=x_0$ because there is no departure. The delivery of the packet under service reduces the AoI to $x_1$ and  thus,  $x'_1=x_1$. Since the arriving source 1 packet is fresh and its age is zero we have $x'_2=0$. Since there is only one packet in the queue,  $x_3$ becomes irrelevant, and  we have  $x'_3=x_3$.
	
		\item \textit{l=5}: A  packet is under service and a source 2 packet arrives. In this transition, we have $x'_0=x_0$ because there is no departure. The delivery of the packet under service reduces the AoI to $x_1$ and  thus,  $x'_1=x_1$. Since the arriving source 2 packet, its delivery does not change the AoI of source 1, and thus we have $x'_2=x_1$. Since there is only one packet in the queue,  $x_3$ becomes irrelevant, and  we have  $x'_3=x_3$.

		\item \textit{l=6}: A source 1 packet is in the queue, a packet is under service and it completes service and is delivered
		to the sink.  With this transition, the AoI at the sink is reset to the age of the  packet that just completed service, and thus, $x'_0=x_1$.  Since the source 1 packet in the queue goes to the server, we have $x'_1=x_2$. In addition, since with this departure the queue becomes empty, we have  $x'_2=x_2$ and $x'_3=x_3$.			 
	
		\item \textit{l=7}: A source 2 packet is in the queue, a packet is under service and it completes service and is delivered
	to the sink.  With this transition, the AoI at the sink is reset to the age of the  packet that just completed service, and thus, $x'_0=x_1$.  Since the source 2 packet in the queue goes to the server and its delivery does not change the AoI of source 1, we have $x'_1=x_1$. In addition, since with this departure the queue becomes empty, we have  $x'_2=x_2$ and $x'_3=x_3$.

	\item \textit{l=8}:  A  packet is under service, a source 1 packet is in the queue, and 
	a  source 1 packet arrives. According to Policy 1, the source 1 packet in the queue is replaced by the fresh
	source 1 packet. In this transition, we have $x'_0=x_0$ because there is no departure. The delivery of the packet under service reduces the AoI to $x_1$, and thus,  $x'_1=x_1$. Since the arriving source 1 packet is fresh and its age is zero, we have $x_2=0$. Since there is only one packet in the queue,  
	we have  $x'_3=x_3$.

	\item \textit{l=9}: 
	A packet is under service, a source 1 packet is in the queue, and 
	a  source 2 packet arrives. In this transition,  ${x'_0=x_0}$ because there is no departure. The delivery of the packet under service reduces the AoI to $x_1$, and thus,  $x'_1=x_1$.  The delivery of the first packet in the queue reduces the AoI to $x_2$, and thus,  ${x'_2=x_2}$. Since the second packet in the queue is a source 2 packet, its delivery does not change the AoI of source 1, and thus we have $x'_3=x_2$. 
	
	\item \textit{l=10}: 
	A packet is under service, a source 2 packet is in the queue, and 
	a  source 1 packet arrives. In this transition,  ${x'_0=x_0}$ because there is no departure. The delivery of the packet under service reduces the AoI to $x_1$, and thus,  $x'_1=x_1$.  Since the first packet in the queue is a source 2 packet, its delivery does not change the AoI of source 1, and thus we have
	$x'_1=x_1$. Since the arriving source 1 packet is fresh and its age is zero, we have $x_3=0$.

	\item \textit{l=11}: A packet is under service, the first packet in the queue is a source 1 packet, the second packet in the queue is a source 2 packet,  and a  source 1 packet arrives. According to Policy 1, the source 1 packet in the queue is replaced by the fresh source 1 packet. In this transition, we have $x'_0=x_0$ because there is no departure. The delivery of the packet under service reduces the AoI to $x_1$, thus,  $x'_1=x_1$. Since the arriving source 1 packet is fresh and its age is zero we have $x'_2=0$. Since the second packet in the queue is a source 2 packet, its delivery does not change the AoI of source 1, and thus
	we have  $x'_3=0$.  
	 The reset maps of transition $l=12$ can be derived similarly.

	\item \textit{l=13}: 
	The first packet in the queue is a source 1 packet, the second packet in the queue is a source 2 packet, and the  packet under service  completes service and is delivered to the sink.
	With this transition, the AoI at the sink is reset to the age of the source 1 packet that just completed service, and thus, $x'_0=x_1$.  Since the first packet in the  queue goes to the server, we have $x'_1=x_2$. 
	In addition, since with this departure the queue holds the  source 2 packet and  its delivery does not change the AoI of source 1, we have $x'_2=x_2$  and $x'_3=x_3$. The reset maps of transition $l=14$  can be derived similarly.

\end{itemize}

Having defined the sets of incoming and outgoing transitions, and the value of $\bar{\bold{v}}_{q_l}\bold{A}_l$ for each incoming transition for each state $q\in\mathcal{Q}$, the remaining task is to derive  $\bold{b}_q, \forall  q\in\mathcal{Q}$, and  the stationary probability vector  $\bar{\boldsymbol{\pi}}$. This is carried out next. 
\subsubsection{Calculation of $\bold{b}_q$ and $\bar\pi_q$ for each state $  q\in\mathcal{Q}$}
The evolution of $ \bold{x}(t) $ at each discrete state $q(t)=q$ is determined by the differential equation   $\dot{\bold{x}}=\bold{b}_q$, as described in Section \ref{Introduction to the SHS Technique}. Since as long as the discrete  state $q(t)$ is unchanged, the age of each element $x_j(t), j\in\{0,\ldots,3\},$ increases at a unit rate with time, and thus  we have $\bold{b}_q=\bold{1}$, where $\bold{1}$ is the row vector $[1 \cdots 1]\in\mathbb{R}^{1\times(n+1)}$.

To calculate the stationary probability vector $\bar{\boldsymbol{\pi}}$, we  use \eqref{eqrt01} and \eqref{erwq}.  Using \eqref{eqrt01} and the transitions between the different states presented in Table \ref{table-2}, it can be shown that the stationary probability vector  $\bar{\boldsymbol{\pi}}$ satisfies 
$
\bar{\boldsymbol{\pi}}\bold{D}=\bar{\boldsymbol{\pi}}\bold{Q} 
$
where the diagonal matrix $\bold{D}\in\mathbb{R}^{(n+1)\times(n+1)}$ and matrix $\bold{Q}\in\mathbb{R}^{(n+1)\times(n+1)}$ are given as
\begin{align}\nonumber
\bold{D}=&\,\,\text{diag}[\lambda,\lambda+\mu,\lambda+\mu,\lambda_1+\mu,\lambda_1+\mu,\lambda_1+\mu],\\&\nonumber
\bold{Q}=\left[\begin{array}{cccccc}
\mu   & \lambda    & 0             &0        & 0        &0\\
0     & 0          & \lambda_1     &\lambda_2& \lambda_2&0 \\
0     & \mu        & \lambda_1     &0        & 0        &0\\
0     & \mu        & 0             &0        & 0        & \lambda_1\\
0     & 0          & 0             &\mu      & \lambda_1& 0\\
0     & 0          & \mu           &0        & 0        & \lambda_1\\
\end{array}\right],
\end{align} 
where $\text{diag}[a_1, a_2,\ldots,a_n]$ denotes a diagonal matrix with elements $a_1, a_2,\ldots,a_n$ on its main diagonal. Using the above $
\bar{\boldsymbol{\pi}}\bold{D}=\bar{\boldsymbol{\pi}}\bold{Q} 
$ and  $\textstyle\sum_{q\in\mathcal{Q}}\bar{{\pi}}_q=1$ in \eqref{erwq}, the stationary probabilities are given as 
\begin{align}\label{proeqq}
\bar{\boldsymbol{\pi}}= \dfrac{1}{\rho^2+\rho(2\rho_1\rho_2+1)+1}\left[1~~ \rho~~ \rho_1\rho ~~\rho_2\rho ~~\rho_1\rho_2\rho~~\rho_1\rho_2\rho\right].
\end{align}
\subsubsection{Average AoI Calculation}
By substituting  \eqref{proeqq} into  \eqref{asleq} and solving the corresponding system of linear equations, the values of $\bar{v}_{q0}, \,\,\forall q\in\mathcal{Q}$, are calculated as presented in Appendix \ref{Valuesof v appendix}. Finally,  substituting the values of $\bar{v}_{q0}, \,\,\forall q\in\mathcal{Q}$, 
into \eqref{AOIANAL}  results in the  average AoI of source 1 under Policy 1, given in Theorem \ref{theo1}.
Note that the expression is \emph{exact}; it characterizes the average AoI in the considered queueing model in \emph{closed form}.

\subsection{Average AoI under Policy 2}\label{Average AoI Under Policy 2}
Recall from Section~\ref{System Model} that the main difference of Policy 2 compared to Policy 1 treated above is that \emph{the system} can contain only two packets, one packet of source 1 and one packet of source 2. Accordingly  for Policy 2, the state space of the Markov chain is $\mathcal{Q}=\{0,1,2,3,4\}$, where
${q=0}$ indicates that the server is idle, i.e., the system is empty; ${q=1}$ indicates that a source 1  packet is under service and the queue is empty; ${q=2}$ indicates that a source 2  packet is under service and the queue is empty; ${q=3}$ indicates that a source 1  packet is under service, and a source 2 packet is in the queue; and  ${q=4}$ indicates that a source 2  packet is under service, and a source 1 packet is in the queue.

The continuous process is ${\bold{x}(t)=[x_0(t)~x_1(t)~x_2(t)]}$, where $x_0(t)$ is the current AoI of source 1 at time instant $t$, $\Delta_1(t)$;  $ x_1(t)  $ encodes  what $\Delta_1(t)$ would become if the packet that is under service is delivered to the sink at time instant $t$; $ x_2(t)  $ encodes  what $\Delta_1(t)$ would become if the  packet in the queue is delivered to the sink at time instant $t$.  Next, we will determine the required quantities  to form the system of linear equations in \eqref{asleq} under Policy 2. 
\subsubsection{Determining the value of $\bar{\bold{v}}_{q_l}\bold{A}_l$ for incoming transitions for each state $q\in\mathcal{Q}$}
The Markov chain for the discrete state $q(t)$ is shown in Fig. \ref{Chain3}.  The transitions between the discrete states ${{q_l \rightarrow q'_l}, \,\,\forall l\in \mathcal{L}}$, and their effects on the continuous state $\bold{x}(t)$ are summarized in Table \ref{table-4}. In the following, we explain the transitions presented in Table \ref{table-4}:

\begin{figure}
	\centering
	\includegraphics[scale=.7]{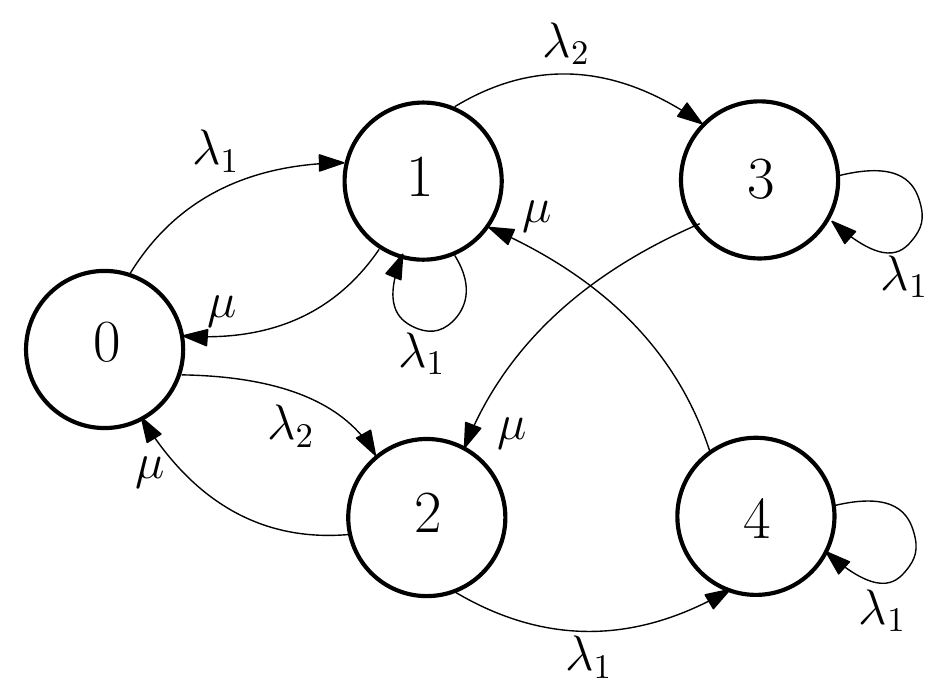}
	\caption{The SHS Markov chain for Policy 2.}
	\vspace{-5mm}
	\label{Chain3}
\end{figure} 

\begin{table}
	\centering\small
	\caption{Table of transitions for the Markov chain of Policy 2 in Fig. \ref{Chain3}}
	\label{table-4}
	\begin{tabular}{ |l|l|c|c|c|c|}
		\hline
		\textit{l}  & $q_l \rightarrow q'_l $&$\lambda^{(l)}$& $\bold{x}\bold{A}_l$&$\bold{A}_l$&$\bold{v}_{q_l}\bold{A}_l$ \\			
		\hline			
		1&$0 \rightarrow  1$& $\lambda_1$&$\left[x_0 ~ 0 ~ x_2 \right]$&$\tiny\begin{bmatrix}
		1 & 0 & 0\\
		0 & 0 & 0\\
		0 & 0 & 1
		\end{bmatrix}$&$\left[v_{00}~ 0~ v_{02} \right]$	\\		
		\hline
		2&$0 \rightarrow  2$&$\lambda_2$&$\left[x_0~ x_0~ x_2 \right]$&$\tiny\begin{bmatrix}
		1 & 1 & 0\\
		0 & 0 & 0\\
		0 & 0 & 1
		\end{bmatrix}$&$\left[v_{00}~ v_{00}~ v_{02} \right]$\\
		\hline	
		3&$1 \rightarrow  1$&$\lambda_1$&$\left[x_0~ 0~ x_2 \right]$&$\tiny\begin{bmatrix}
		1 & 0 & 0\\
		0 & 0 & 0\\
		0 & 0 & 1
		\end{bmatrix}$&$\left[v_{10}~ 0~ v_{12} \right]$\\
		\hline	
		4&$1 \rightarrow  3$&$\lambda_2$&$\left[x_0~ x_1~ x_1 \right]$&$\tiny\begin{bmatrix}
		1 & 0 & 0\\
		0 & 1 & 1\\
		0 & 0 & 0
		\end{bmatrix}$&$\left[v_{10}~ v_{11}~ v_{11} \right]$\\
		\hline	
		5&$2 \rightarrow 4 $&$\lambda_1$&$\left[x_0~ x_0~ 0 \right]$&$\tiny\begin{bmatrix}
		1 & 1 & 0\\
		0 & 0 & 0\\
		0 & 0 & 0
		\end{bmatrix}$&$\left[v_{20}~ v_{20}~ 0 \right]$\\
		\hline	
		6&$3 \rightarrow  3$&$\lambda_1$&$\left[x_0~ 0~ 0 \right]$&$\tiny\begin{bmatrix}
		1 & 0 & 0\\
		0 & 0 & 0\\
		0 & 0 & 0
		\end{bmatrix}$&$\left[v_{30}~ 0~ 0 \right]$\\
		\hline	
		7&$4 \rightarrow 4 $&$\lambda_1$&$\left[x_0~ x_0~ 0 \right]$&$\tiny\begin{bmatrix}
		1 & 1 & 0\\
		0 & 0 & 0\\
		0 & 0 & 0
		\end{bmatrix}$&$\left[v_{40}~ v_{40}~ 0 \right]$\\
		\hline	
		8&$1 \rightarrow  0$&$\mu$&$\left[x_1~ x_1~ x_2 \right]$&$\tiny\begin{bmatrix}
		0 & 0 & 0\\
		1 & 1 & 0\\
		0 & 0 & 1
		\end{bmatrix}$&$\left[v_{11}~ v_{11}~ v_{12} \right]$\\
		\hline	
		9&$2 \rightarrow  0$&$\mu$&$\left[x_0~ x_1~ x_2 \right]$&$\tiny\begin{bmatrix}
		1 & 0 & 0\\
		0 & 1 & 0\\
		0 & 0 & 1
		\end{bmatrix}$&$\left[v_{20}~ v_{21}~ v_{22} \right]$\\
		\hline	
		10&$3 \rightarrow  2$&$\mu$&$\left[x_1~ x_1~ x_2 \right]$&$\tiny\begin{bmatrix}
		0 & 0 & 0\\
		1 & 1 & 0\\
		0 & 0 & 1
		\end{bmatrix}$&$\left[v_{31}~ v_{31}~ v_{32} \right]$\\
		\hline	
		11&$4 \rightarrow  1$&$\mu$&$\left[x_0~ x_2~ x_2 \right]$&$\tiny\begin{bmatrix}
		1 & 0 & 0\\
		0 & 0 & 0\\
		0 & 1 & 1
		\end{bmatrix}$&$\left[v_{40}~ v_{42}~ v_{42} \right]$\\
		\hline			
	\end{tabular}
\end{table}

\begin{itemize}
	\item \textit{l=1}: A source 1 packet  arrives at an empty system. With this transition we have  $x'_0=x_0$ because there is no departure.  Since with this arrival the queue is still empty, $x_2$ becomes irrelevant to the AoI of source 1, and thus, $x'_2=x_2$.

	\item \textit{l=2}: A  source 2 packet arrives at an empty system.  We have $x'_0=x_0$, because this arrival does not change the AoI at the sink.  Since the arriving packet is a source 2 packet, its delivery does not change the AoI of source 1, and thus  we have
	$x'_1=x_0$. Moreover, since the queue is empty, $x_2$ becomes irrelevant, and thus, we have  $x'_2=x_2$.

	\item \textit{l=3}: A source 1 packet is under service and a source 1 packet arrives. According to the self-preemptive service of Policy 2, the source 1 packet that is under service is preempted by the arriving source 1 packet.
	In this transition, we have $x'_0=x_0$ because there is no departure. 
	Since the arrived source 1 packet that entered the server through the preemption is fresh and its age is zero, we have $x'_1=0$. Since the queue is empty, $x_2$ becomes irrelevant, and thus, we have  $x'_2=x_2$.

	\item \textit{l=4}: A source 1 packet is under service and a source 2 packet arrives. In this transition, we have $x'_0=x_0$ because there is no departure. The delivery of the packet under service reduces the AoI to $x_1$, and thus,  we have $x'_1=x_1$. Since the packet in the queue is a source 2 packet, its delivery does not change the AoI of source 1, and  thus
	we have $x'_2=x_1$.	 
	The reset map of transition $l=6$ can be derived similarly.

	\item \textit{l=5}: A source 2 packet is under service and a source 1 packet arrives. In this transition, we have $x'_0=x_0$ because there is no departure. Since the packet under service is a source 2 packet, its delivery does not change the AoI of source 1, and thus we have $x'_1=x_0$. Since the arriving source 1 packet is fresh and its age is zero, we have $x_2=0$.

	\item \textit{l=6}: A source 1 packet is under service, the  packet in the queue is a source 2 packet,  and a  source 1 packet arrives. According to the self-preemptive policy, the source 1 packet that is under service is preempted by the arriving source 1 packet.
	In this transition, we have $x'_0=x_0$ because there is no departure. 	Since the arrived source 1 packet that entered the server through the preemption is fresh and its age is zero, we have $x'_1=0$.
	Since the  packet in the queue is a source 2 packet, its delivery does not change the AoI of source 1, and thus
	we have  $x'_2=0$.   
	The reset maps of transition $l=7$ can be derived similarly.

	\item \textit{l=8}: A source 1 packet is under service and it completes service and is delivered
	to the sink. With this transition, 
	the AoI at the sink is reset to the age of the source 1 packet that just completed service, and thus, $x'_0=x_1$. Since the 
	system enters state $q=0$,
	we have  ${x'_1=x_1}$, and ${x'_2=x_2}$. The reset map of transition $l=9$ can be derived similarly.

	\item \textit{l=10}: 
	The  packet in the queue is a source 2 packet and the source 1 packet in the server completes service and is delivered to the sink.
	With this transition, the AoI at the sink is reset to the age of the source 1 packet that just completed service, i.e., $x'_0=x_1$.  Since the packet that goes to the server is a source 2 packet, its delivery does not change the AoI of source 1, and  thus   we have $x'_1=x_1$. In addition, since with this transition the queue becomes empty, we have $x'_2=x_2$.	
	The reset map of transition $l=11$ can be derived similarly.

\end{itemize}
\subsubsection{Calculation of $\bold{b}_q$ and $\bar\pi_q$ for each state $ q\in\mathcal{Q}$}

Similarly as  in Section \ref{FCFS Prioritized Packet Management Policy},  as long as the discrete  state $q(t)$ is unchanged, the age of each element $x_j(t), j\in\{0,\ldots,2\},$ increases at a unit rate with time. Thus, we have $\bold{b}_q=\bold{1}$.
Next, we calculate the stationary probability vector  $\bar{\boldsymbol{\pi}}$.  Using \eqref{eqrt01} and the transition rates among the different states presented in Table \ref{table-4}, it can be shown that the stationary probability vector  $\bar{\boldsymbol{\pi}}$ satisfies $\bar{\boldsymbol{\pi}}\bold{D}=\bar{\boldsymbol{\pi}}\bold{Q}$ with
\begin{align}\nonumber
\bold{D}=&\,\,\text{diag}[\lambda,\lambda+\mu,\lambda_1+\mu,\lambda_1+\mu,\lambda_1+\mu],\\&\nonumber
\bold{Q}=\left[\begin{array}{ccccc}
0   & \lambda_1 & \lambda_2&0& 0 \\
\mu & \lambda_1 & 0&\lambda_2& 0 \\
\mu & 0 & 0&0& \lambda_1 \\
0   & 0 & \mu&\lambda_1& 0 \\
0   & \mu & 0&0& \lambda_1\\
\end{array}\right]\!.
\end{align}
Using the above $
\bar{\boldsymbol{\pi}}\bold{D}=\bar{\boldsymbol{\pi}}\bold{Q} 
$  and  $\textstyle\sum_{q\in\mathcal{Q}}\bar{{\pi}}_q=1$ in \eqref{erwq}, the stationary probabilities are given as
\begin{align}\label{proeqq0}
\bar{\boldsymbol{\pi}}= \dfrac{1}{2\rho_1\rho_2+\rho+1}\left[1~~ \rho_1~~ \rho_2 ~~\rho_1\rho_2 ~~\rho_1\rho_2\right].
\end{align}

\subsubsection{Average AoI Calculation}
By substituting \eqref{proeqq0} into  \eqref{asleq} and solving the corresponding system of linear equations, the values of $\bar{v}_{q0}, \,\,\forall q\in\mathcal{Q}$, are calculated as
\begin{align}\label{V2AoI2}
&\bar{v}_{00}=\dfrac{\rho_1^2(2\rho+5)+(4\rho_1+1)(\rho_2+1)}{\mu\rho_1(1+\rho_1)^2(1+\rho)(1+\rho+2\rho_1\rho_2)},\\\nonumber&
\bar{v}_{10}=\dfrac{(1+\rho_2)(\rho_1^3+4\rho_1^2+1)+\rho_1(5\rho_2+4)}{\mu(1+\rho_2)(1+\rho_1)^2(1+\rho+2\rho_1\rho_2)},\\\nonumber&
\bar{v}_{20}=\dfrac{\rho_2\big(\rho_1^2(2\rho+6)+(4\rho_1+1)(\rho_2+1)\big)}{\mu\rho_1(1+\rho_1)^2(1+\rho)(1+\rho+2\rho_1\rho_2)},\\\nonumber&
\bar{v}_{30}=\dfrac{\rho_2\big((1+\rho_2)(2\rho_1^3+6\rho_1^2+1)+\rho_1(6\rho_2+5)\big)}{\mu(1+\rho_2)(1+\rho_1)^2(1+\rho+2\rho_1\rho_2)},\\\nonumber&
\bar{v}_{40}=\dfrac{\rho_2\big(\rho_1^2(\rho_1^2+5\rho_1+\rho_1\rho_2+4\rho_2+9)+(5\rho_1+1)(1+\rho_2)\big)}{\mu(1+\rho_1)^2(1+\rho)(1+\rho+2\rho_1\rho_2)}.
\end{align}
 Finally,  substituting the values of $\bar{v}_{q0}, \,\,\forall q\in\mathcal{Q}$, in \eqref{V2AoI2} into \eqref{AOIANAL}  results in the  average AoI of source 1 under Policy 2, given in Theorem \ref{theo2}.

 \subsection{Average AoI under Policy 3}\label{Average AoI Under Policy 3}
 The main difference of Policy 3 compared to Policy 2 is that it does not permit preemption in service. 
  The Markov chain and the continuous process of Policy 3 are the same as those for  Policy 2. Thus, the stationary probability vector  $\bar{\boldsymbol{\pi}}$ of Policy 3 is given in \eqref{proeqq0}.   The transitions between the discrete states ${{q_l \rightarrow q'_l}, \,\,\forall l\in \mathcal{L}}$, and their effects on the continuous state $\bold{x}(t)$ are summarized in Table \ref{table-5}. The reset maps of transitions $ l\in\{1,2,4,5,7,8,9,10,11\}$ are the same as those for Policy 2.  Thus, we only explain  transitions $l=3$ and $l=6$ (see Table \ref{table-5}).  
 \begin{table}
 	\centering\small
 	\caption{Table of transitions for the Markov chain of Policy 3}
 	\label{table-5}
 	\begin{tabular}{ |l|l|c|c|c|c|}
 		\hline
 		\textit{l}  & $q_l \rightarrow q'_l $&$\lambda^{(l)}$& $\bold{x}\bold{A}_l$&$\bold{A}_l$&$\bold{v}_{q_l}\bold{A}_l$ \\			
 		\hline			
 		1&$0 \rightarrow  1$& $\lambda_1$&$\left[x_0 ~ 0 ~ x_2 \right]$&$\tiny\begin{bmatrix}
 		1 & 0 & 0\\
 		0 & 0 & 0\\
 		0 & 0 & 1
 		\end{bmatrix}$&$\left[v_{00}~ 0~ v_{02} \right]$	\\		
 		\hline
 		2&$0 \rightarrow  2$&$\lambda_2$&$\left[x_0~ x_0~ x_2 \right]$&$\tiny\begin{bmatrix}
 		1 & 1 & 0\\
 		0 & 0 & 0\\
 		0 & 0 & 1
 		\end{bmatrix}$&$\left[v_{00}~ v_{00}~ v_{02} \right]$\\
 		\hline	
 		3&$1 \rightarrow  1$&$\lambda_1$&$\left[x_0~ x_1~ x_2 \right]$&$\tiny\begin{bmatrix}
 		1 & 0 & 0\\
 		0 & 1 & 0\\
 		0 & 0 & 1
 		\end{bmatrix}$&$\left[v_{10}~ v_{11}~ v_{12} \right]$\\
 		\hline	
 		4&$1 \rightarrow  3$&$\lambda_2$&$\left[x_0~ x_1~ x_1 \right]$&$\tiny\begin{bmatrix}
 		1 & 0 & 0\\
 		0 & 1 & 1\\
 		0 & 0 & 0
 		\end{bmatrix}$&$\left[v_{10}~ v_{11}~ v_{11} \right]$\\
 		\hline	
 		5&$2 \rightarrow 4 $&$\lambda_1$&$\left[x_0~ x_0~ 0 \right]$&$\tiny\begin{bmatrix}
 		1 & 1 & 0\\
 		0 & 0 & 0\\
 		0 & 0 & 0
 		\end{bmatrix}$&$\left[v_{20}~ v_{20}~ 0 \right]$\\
 		\hline	
 		6&$3 \rightarrow  3$&$\lambda_1$&$\left[x_0~ x_1~ x_1 \right]$&$\tiny\begin{bmatrix}
 		1 & 0 & 0\\
 		0 & 1 & 1\\
 		0 & 0 & 0
 		\end{bmatrix}$&$\left[v_{30}~ v_{31}~ v_{31} \right]$\\
 		\hline	
 		7&$4 \rightarrow 4 $&$\lambda_1$&$\left[x_0~ x_0~ 0 \right]$&$\tiny\begin{bmatrix}
 		1 & 1 & 0\\
 		0 & 0 & 0\\
 		0 & 0 & 0
 		\end{bmatrix}$&$\left[v_{40}~ v_{40}~ 0 \right]$\\
 		\hline	
 		8&$1 \rightarrow  0$&$\mu$&$\left[x_1~ x_1~ x_2 \right]$&$\tiny\begin{bmatrix}
 		0 & 0 & 0\\
 		1 & 1 & 0\\
 		0 & 0 & 1
 		\end{bmatrix}$&$\left[v_{11}~ v_{11}~ v_{12} \right]$\\
 		\hline	
 		9&$2 \rightarrow  0$&$\mu$&$\left[x_0~ x_1~ x_2 \right]$&$\tiny\begin{bmatrix}
 		1 & 0 & 0\\
 		0 & 1 & 0\\
 		0 & 0 & 1
 		\end{bmatrix}$&$\left[v_{20}~ v_{21}~ v_{22} \right]$\\
 		\hline	
 		10&$3 \rightarrow  2$&$\mu$&$\left[x_1~ x_1~ x_2 \right]$&$\tiny\begin{bmatrix}
 		0 & 0 & 0\\
 		1 & 1 & 0\\
 		0 & 0 & 1
 		\end{bmatrix}$&$\left[v_{31}~ v_{31}~ v_{32} \right]$\\
 		\hline	
 		11&$4 \rightarrow  1$&$\mu$&$\left[x_0~ x_2~ x_2 \right]$&$\tiny\begin{bmatrix}
 		1 & 0 & 0\\
 		0 & 0 & 0\\
 		0 & 1 & 1
 		\end{bmatrix}$&$\left[v_{40}~ v_{42}~ v_{42} \right]$\\
 		\hline			
 	\end{tabular}
 \end{table}

\begin{itemize}
 	\item \textit{l=3}: A source 1 packet is under service and a source 1 packet arrives. According to Policy 3, the arrived  source 1 packet is blocked and cleared.
 	In this transition, we have $x'_0=x_0$ because there is no departure. The delivery of the packet under service reduces the AoI to $x_1$, and thus,  we have $x'_1=x_1$.
 	 Since the queue is empty, $x_2$ becomes irrelevant, and thus, we have  $x'_2=x_2$.  
 	\item \textit{l=6}: A source 1 packet is under service, the  packet in the queue is a source 2 packet,  and a  source 1 packet arrives. The arrived source 1 packet is blocked and cleared.
 	In this transition, we have $x'_0=x_0$ because there is no departure. 	The delivery of the packet under service reduces the AoI to $x_1$, and thus,  we have $x'_1=x_1$.
 	Since the  packet in the queue is a source 2 packet, its delivery does not change the AoI of source 1, and thus
 	we have  $x'_2=x_1$.   
 \end{itemize}

Having  the stationary probability vector  $\bar{\boldsymbol{\pi}}$ (given in \eqref{proeqq0}) and the table of transitions (Table \ref{table-5}), we can form the system of linear equations \eqref{asleq}. By 
 solving the  system of linear equations, the values of $\bar{v}_{q0}, \,\,\forall q\in\mathcal{Q}$, are calculated as
 \begin{align}\label{V2AoI20}
 &\bar{v}_{00}=\dfrac{\rho_1^3+\rho_1^2((\rho_2+2)^2-1)+(\rho_2+1)^2(3\rho_1+1)}{\mu\rho_1(1+\rho_1)(1+\rho_2)(1+\rho)(1+\rho+2\rho_1\rho_2)},\\\nonumber&
 \bar{v}_{10}=\dfrac{(1+\rho_2)(2\rho_1^2+1)+\rho_1(4\rho_2+3)}{\mu(1+\rho_2)(1+\rho_1)(1+\rho+2\rho_1\rho_2)},\\\nonumber&
 \bar{v}_{20}=\dfrac{\rho_2\big(\rho_1^3(\rho_2+2)+\rho_1^2(\rho_2^2+5\rho_2+4)+(3\rho_1+1)(\rho_2+1)^2\big)}{\mu\rho_1(1+\rho_1)(1+\rho_2)(1+\rho)(1+\rho+2\rho_1\rho_2)},\\\nonumber&
 \bar{v}_{30}=\dfrac{\rho_2\big((\rho_2+1)(3\rho_1^2+1)+\rho_1(5\rho_2+4)\big)}{\mu(1+\rho_1)(1+\rho_2)(1+\rho+2\rho_1\rho_2)},\\\nonumber&
 \bar{v}_{40}=\dfrac{\rho_2\big(\rho_1^3(2\rho_2+3)+2\rho_1^2((\rho_2+2)^2-1)+(4\rho_1+1)(\rho_2+1)^2\big)}{\mu(1+\rho_1)(1+\rho_2)(1+\rho)(1+\rho+2\rho_1\rho_2)}.
 \end{align}
 Finally,  substituting the values of $\bar{v}_{q0}, \,\,\forall q\in\mathcal{Q}$, in \eqref{V2AoI20} into \eqref{AOIANAL}  results in the  average AoI of source 1 under Policy 3, given in Theorem \ref{theo3}.

\section{Numerical Results}\label{Numerical Results}
In this section, we show the effectiveness of the proposed packet management policies in terms of the sum average AoI and fairness between the different sources in the system. Moreover, we compare our policies against the following existing policies: 
 the source-agnostic packet management policies LCFS-S and LCFS-W proposed in \cite{8469047}, and the priority based packet management policies  proposed in \cite{8437591}, which we term PP-NW and PP-WW. Under the LCFS-S policy, a new arriving packet preempts any packet that is currently under service (regardless of the source index). Under the LCFS-W policy, a new arriving packet replaces any older packet waiting in the queue (regardless of the source index); however, the new packet has to wait for any packet under service to finish. Under the  PP-NW policy, there is no waiting room
 and an update under service is preempted on arrival of an equal or
 higher priority update. Under the PP-WW policy, there is a waiting room for at
 most one update and preemption is allowed in waiting but not in
 service. 
 Without loss of generality, for the PP-NW and PP-WW policies, we assume that source 2 has higher priority than source 1; for the opposite case, the results are symmetric.

\subsubsection{Average AoI}
Fig. \ref{1Comparision} depicts the contours of achievable average AoI  pairs
$(\Delta_1,\Delta_2)$ for fixed values of system load $\rho=\rho_1+\rho_2$ under different packet management policies with normalized service rate $\mu=1$; in Fig.~\ref{1Comparision}(a),  ${\rho=1}$   and in Fig.~\ref{1Comparision}(b), ${\rho=6}$. This figure shows that under an appropriate packet management policy in the system (either in the queue or server), by increasing the load of the system the average AoI decreases.
Besides that, it shows that Policy 2 provides the lowest average AoI as compared to the other policies. 

\begin{figure}
	\centering
	\subfloat[]
	{\includegraphics[width=0.54\linewidth,trim = 10mm 0mm 20mm 10mm,clip]{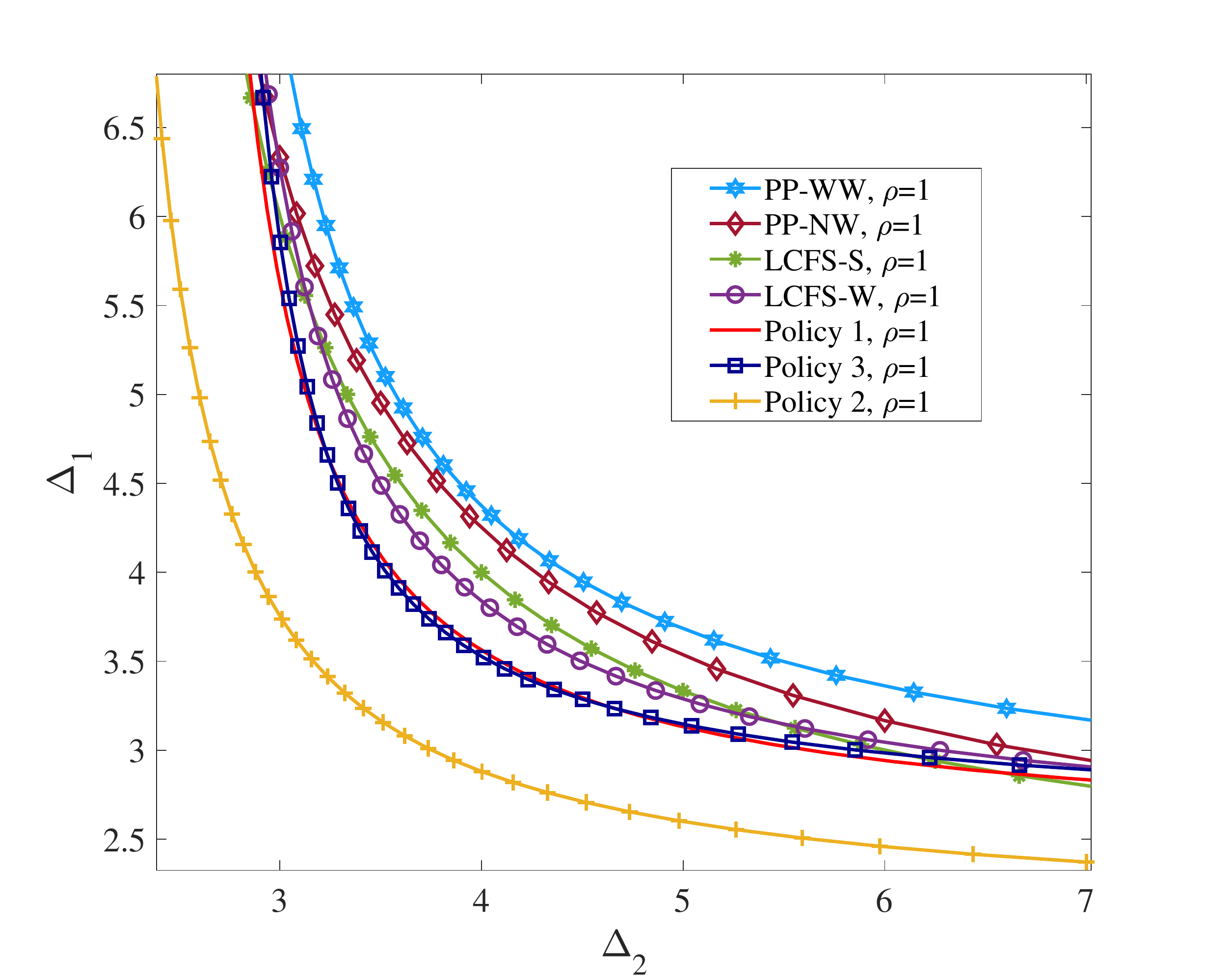}}\\	
	\subfloat[]
	{\includegraphics[width=0.55\linewidth,trim = 10mm 0mm 20mm 10mm,clip]{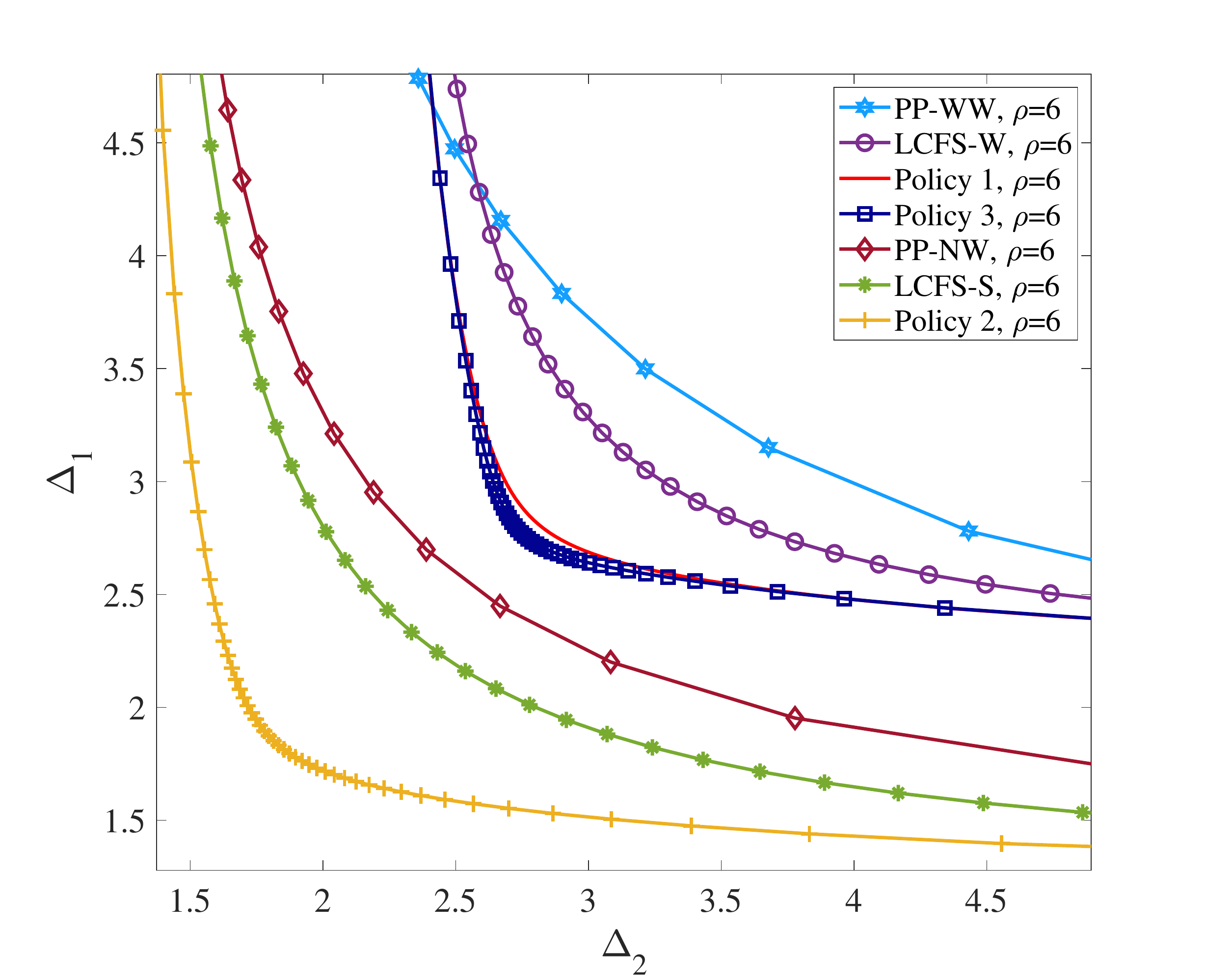}}\\
	\caption{ The average AoI of  sources 1  and 2 under different packet management policies for $\mu=1$ with (a)  $\rho=\rho_1+\rho_2=1$, and  (b)  $\rho=\rho_1+\rho_2=6$.}
	\vspace{-12mm}
	\label{1Comparision}
\end{figure}

\subsubsection{Sum Average AoI}
  Fig. \ref{2Comparision} depicts sum average AoI as a function of $\rho_1$ under different packet management policies with $\mu=1$; in Fig. \ref{2Comparision}(a),  ${\rho=\rho_1+\rho_2=1}$   and in Fig. \ref{2Comparision}(b), ${\rho=\rho_1+\rho_2=6}$. 
  This figure shows that Policy 2 provides the lowest average AoI for all values of $ \rho_1 $ as compared to the other policies. 
  In addition, we can observe that among Policy 1, Policy 3, PP-NW, PP-WW, LCFS-S, and LCFS-W policies, the policy that achieves the lowest value of the sum average AoI depends on the system parameters. Moreover, we can observe that under the PP-NW and PP-WW policies the minimum value of sum average AoI is achieved for a high value of $\rho_1$. This is because when priority is with source 2, a high value of $\rho_1$ is needed to compensate for the priority. In addition, we can see that for a  high value of total load, i.e., $ \rho=6 $, the range of values of $ \rho_1 $ for which  PP-NW and PP-WW policies operate well becomes narrow. 
\begin{figure}
	\centering
	\subfloat[]
	{\includegraphics[width=0.53\linewidth,trim = 10mm 0mm 20mm 10mm,clip]{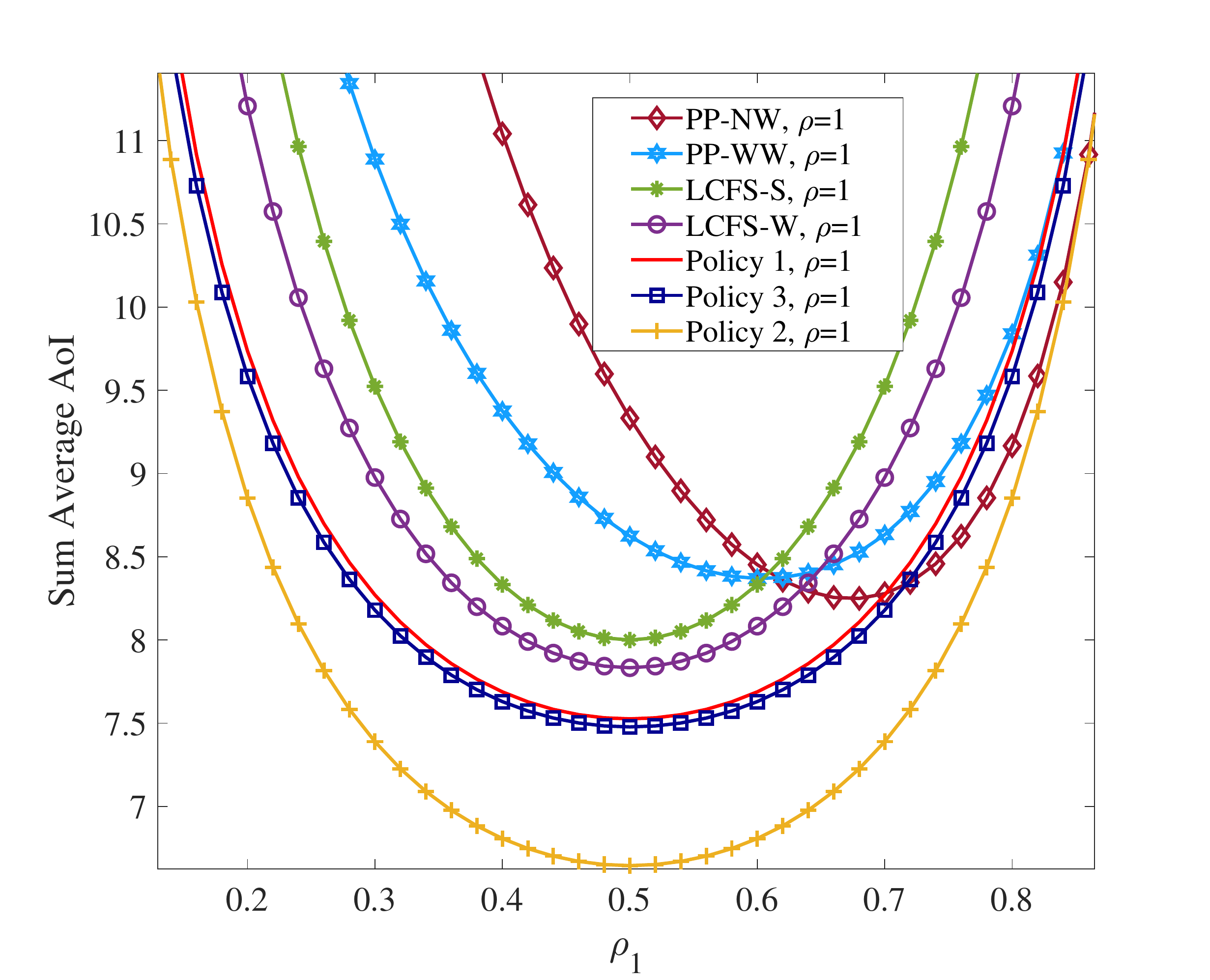}}\\	
	\subfloat[]
	{\includegraphics[width=0.54\linewidth,trim = 10mm 0mm 20mm 10mm,clip]{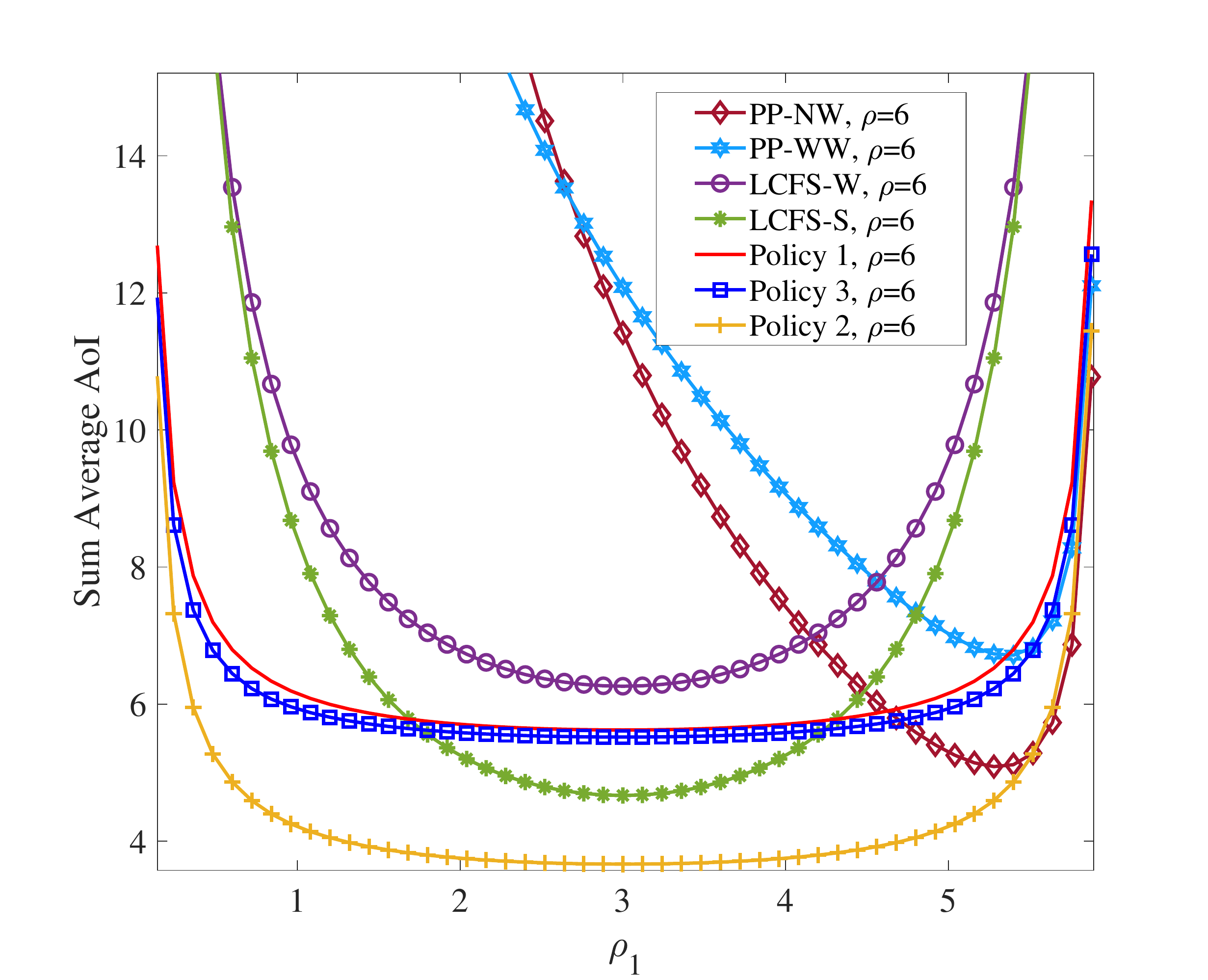}}\\
	\caption{Sum average AoI as a function of $\rho_1$ under different packet management policies for $\mu=1$ with (a)  $\rho=\rho_1+\rho_2=1$, and  (b)  $\rho=\rho_1+\rho_2=6$.}
	\vspace{-12mm}
	\label{2Comparision}
\end{figure}

\subsubsection{Fairness}
For some applications, besides the sum average AoI,  the individual average AoI of each source is critical. In this case, fairness between different sources becomes important to be taken into account.  To compare the fairness between different sources under the different packet management policies, we use the Jain's fairness index \cite{6547814}. For the average AoI of sources 1 and 2, the Jain's fairness index $J(\Delta_1,\Delta_2)$ is defined as \cite[Definition~1]{6547814} \cite[Section~3]{8901143halim} 
\begin{align}\label{01mnn}
J(\Delta_1,\Delta_2)=\dfrac{(\Delta_1+\Delta_2)^2}{2(\Delta_1^2+\Delta_2^2)}.
\end{align}
The Jain's index  $J(\Delta_1,\Delta_2)$ is continuous and lies in $[0.5,1]$, where
$  J(\Delta_1,\Delta_2) = 1 $ indicates the fairest situation in the system.
\begin{figure}
	\centering
	\subfloat[]
	{\includegraphics[width=0.55\linewidth,trim = 10mm 0mm 20mm 15mm,clip]{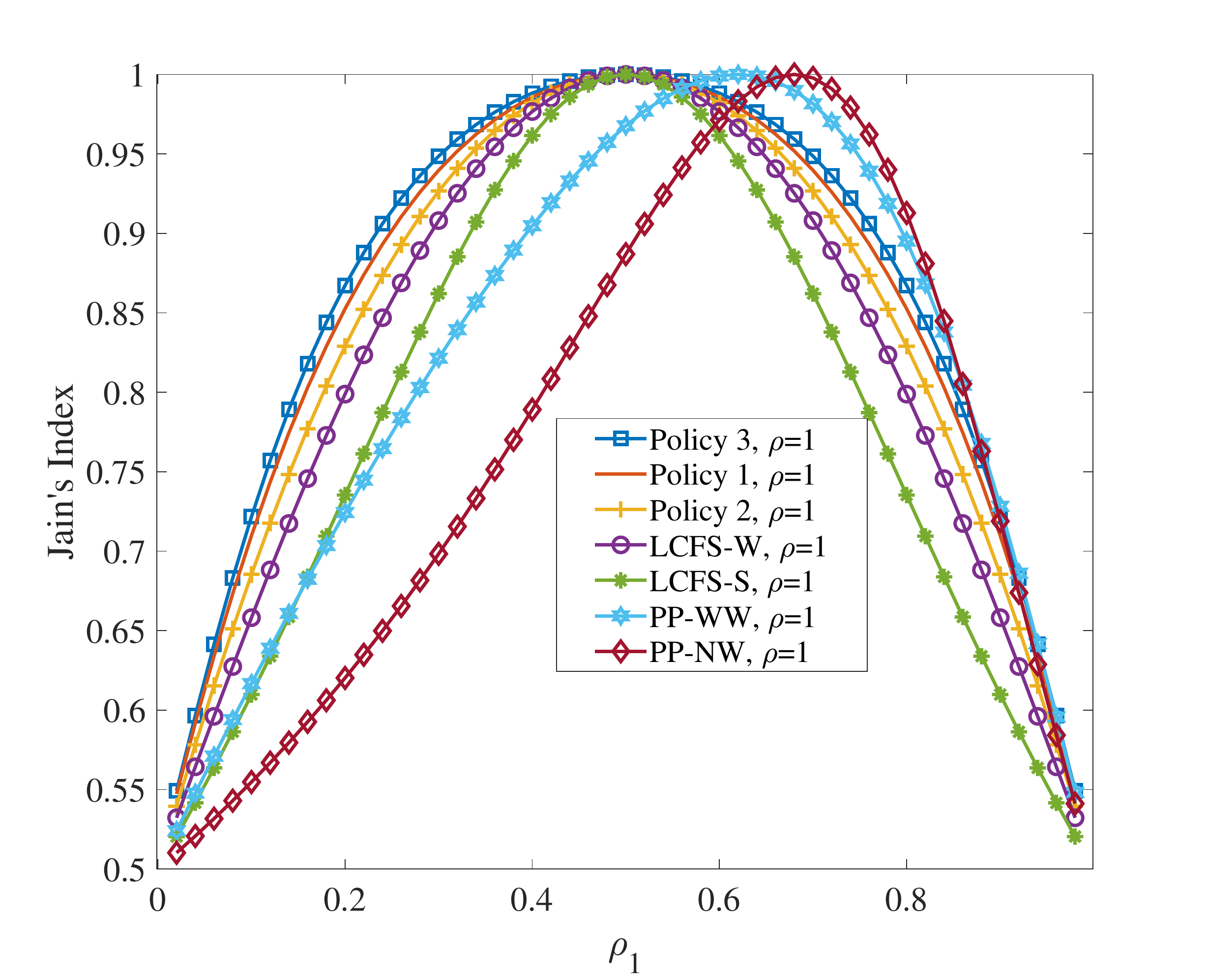}}\\
	\subfloat[]
	{\includegraphics[width=0.55\linewidth,trim = 10mm 0mm 20mm 10mm,clip]{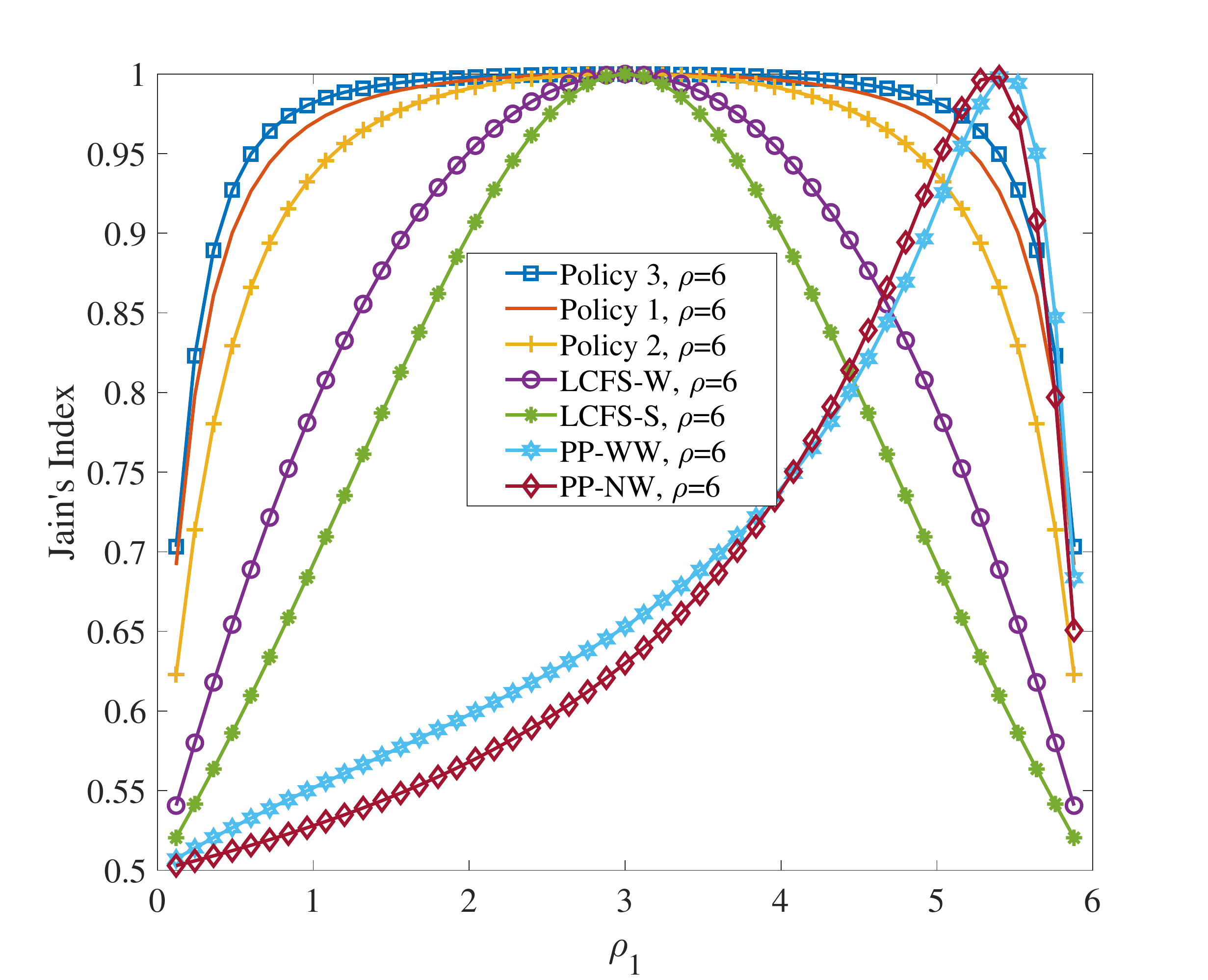}}\\	
	\caption{The Jain's fairness index for  the average AoI of sources 1  and 2 as a function of $\rho_1$ under different packet management policies for  $\mu=1$ with (a)  $\rho=\rho_1+\rho_2=1$, and  (b)  $\rho=\rho_1+\rho_2=6$.}
	\vspace{-12mm}
	\label{FairnessComparision}
\end{figure}

Fig. \ref{FairnessComparision} depicts the Jain's fairness index for the average AoI of sources 1  and 2 as a function of $\rho_1$ under different packet management policies with $\mu=1$;  in Fig. \ref{FairnessComparision}(a),  ${\rho=\rho_1+\rho_2=1}$  and in Fig. \ref{FairnessComparision}(b),  ${\rho=\rho_1+\rho_2=6}$.
As it can be seen, among Policy 1, Policy 2, Policy 3, LCFS-S, and LCFS-W policies, the LCFS-S policy  provides the lowest fairness in the system. This is because 
the packets of a source with a lower packet arrival rate are most of the time  preempted by the packets of the other source having a higher packet arrival rate. In addition, we  observe that the proposed source-aware policies (i.e., Policy 1, Policy 2, and Policy 3) in general provide better fairness than that of the other policies and Policy 3 provides the fairest situation in the system. This is because under these policies, a packet in the queue or server can be preempted  \textit{only} by a packet with the same source index. Similarly as in Fig. \ref{2Comparision}(b), for the high load case, the range of values of $ \rho_1 $ for which  PP-NW and PP-WW policies provide a good fairness becomes narrow.

\section{Conclusions}\label{Conclusions}
We considered a status update system consisting of two independent sources, one server, and one  sink. We proposed three source-aware packet management policies where differently from the existing works, a packet in the system can be preempted  only by a packet with the same source index.  We derived the average AoI for each source under the proposed packet management policies using the SHS technique. 
 The numerical results showed that Policy 2 results in a lower sum average  AoI in the system compared to the  existing policies. In addition, the experiments showed that in general the proposed source-aware policies result in higher fairness in the system than that of the existing policies and, in particular, Policy 3 provides the fairest situation in the system.
 
  The present paper opens several research directions for future study. It would be interesting to extend the results for more than two sources. The same methodology as used in this paper can be applied for more than two sources; however, the complexity of the calculations increases exponentially with the number of sources. Another interesting future work includes calculating the stationary distribution of the AoI under the proposed packet management policies.

\appendices
\section{Values of $\bar{v}_{q0}, \,\,\forall q\in\mathcal{Q},$ for Policy 1} \label{Valuesof v appendix}

\begin{align}\label{VAoI0}
\bar{v}_{00}=\dfrac{3\rho_1^4+\rho_1^3(5\rho_2+9)+\rho_1^2(2\rho_2^2+11\rho_2+10)+\rho_1(4\rho_2^2+6\rho_2+5)+\rho_2^2+\rho_2+1}{\mu\rho_1(1+\rho_1)^2\big((\rho+1)^2-\rho_2\big)\big(\rho^2+\rho(2\rho_1\rho_2+1)+1\big)},
\end{align}
\begin{align}\label{VAoI1}
\bar{v}_{10}=\dfrac{\rho_2^5+3\rho_2^4+4\rho_2^3+3\rho_2^2+\rho_2+\sum_{k=1}^{7}\rho_1^k\gamma_{1,k}}{\mu\rho_1(1+\rho)(1+\rho_2)(1+\rho_1)^2\big((\rho+1)^2-\rho_2\big)\big(\rho^2+\rho(2\rho_1\rho_2+1)+1\big)},
\end{align}
where
\begin{align}\nonumber
&\gamma_{1,1}=4\rho_2^5+16\rho_2^4+26\rho_2^3+22\rho_2^2+9\rho_2+1,~~~~
\gamma_{1,2}= 2\rho_2^5+25\rho_2^4+64\rho_2^3+70\rho_2^2+35\rho_2+6,\\\nonumber&
\gamma_{1,3}=11\rho_2^4+62\rho_2^3+107\rho_2^2+72\rho_2+16,~~~~
\gamma_{1,4}= \rho_2^4+23\rho_2^3+75\rho_2^2+77\rho_2+23,\\\nonumber&
\gamma_{1,5}=3\rho_2^3+23\rho_2^2+41\rho_2+18,~~~~~~~~~~~~~~~~
\gamma_{1,6}=3\rho_2^2+10\rho_2+7,
~~~
\gamma_{1,7}=\rho^2_2+10\rho_2+1.
\end{align}
\begin{align}\label{VAoI22}
\bar{v}_{20}=\dfrac{\rho_2^6+4\rho_2^5+7\rho_2^4+7\rho_2^3+4\rho_2^2+\rho_2+\sum_{k=1}^{7}\rho_1^k\gamma_{2,k}}{\mu(1+\rho)(1+\rho_1)^2(1+\rho_2)^2\big((\rho+1)^2-\rho_2\big)\big(\rho^2+\rho(2\rho_1\rho_2+1)+1\big)},
\end{align}
where
\begin{align}\nonumber
&\gamma_{2,1}=5\rho_2^6+23\rho_2^5+46\rho_2^4+51\rho_2^3+32\rho_2^2+10\rho_2+1,~~~~\\\nonumber&
\gamma_{2,2}=4\rho_2^6+36\rho_2^5+108\rho_2^4+156\rho_2^3+119\rho_2^2+46\rho_2+7,\\\nonumber&
\gamma_{2,3}=\rho_2^6+2\rho_2^5+100\rho_2^4+213\rho_2^3+222\rho_2^2+111\rho_2+21,~~~~~~\\\nonumber&
\gamma_{2,4}=4\rho_2^5+40\rho_2^4+134\rho_2^3+202\rho_2^2+138+33,~~~~~~
\gamma_{2,5}=6\rho_2^4+39\rho_2^3+89\rho_2^2+87\rho_2+28,\\\nonumber&
\gamma_{2,6}=4\rho_2^3+18\rho_2^2+26\rho_2+12,~~~~~~~~~~~~~~~~~~~~~~~~~~
\gamma_{2,7}=\rho_2^2+3\rho_2+2.
\end{align}
\begin{align}\label{VAoI3}
\bar{v}_{30}=\dfrac{\rho_2^6+3\rho_2^5+4\rho_2^4+3\rho_2^3+\rho_2^2+\sum_{k=1}^{7}\rho_1^k\gamma_{3,k}}{\mu\rho_1(1+\rho)(1+\rho_2)(1+\rho_1)^2\big((\rho+1)^2-\rho_2\big)\big(\rho^2+\rho(2\rho_1\rho_2+1)+1\big)},
\end{align}
where
\begin{align}\nonumber
&\gamma_{3,1}=5\rho_2^6+19\rho_2^5+30\rho_2^4+25\rho_2^3+10\rho_2^2+\rho_2,~
\gamma_{3,2}= 5\rho_2^6+37\rho_2^5+84\rho_2^4+87\rho_2^3+41\rho_2^2+6\rho_2,\\\nonumber&
\gamma_{3,3}=2\rho_2^6+27\rho_2^5+101\rho_2^4+149\rho_2^3+91\rho_2^2+18\rho_2,~
\gamma_{3,4}=8\rho_2^5+55\rho_2^4+126\rho_2^3+110\rho_2^2+30\rho_2,\\\nonumber&
\gamma_{3,5}=12\rho_2^4+51\rho_2^3+69\rho_2^2+27\rho_2,~~~~~~~~~~~~
\gamma_{3,6}=8\rho_2^3+20\rho_2^2+12\rho_2,~~~~~~~~
\gamma_{3,7}=2\rho_2^2+2\rho_2.
\end{align}
\begin{align}\label{VAoI4}
\bar{v}_{40}=\dfrac{\rho_2^7+4\rho_2^6+7\rho_2^5+7\rho_2^4+4\rho_2^3+\rho_2^2+\sum_{k=1}^{7}\rho_1^k\gamma_{4,k}}{\mu(1+\rho)(1+\rho_2)^2(1+\rho_1)^2\big((\rho+1)^2-\rho_2\big)\big(\rho^2+\rho(2\rho_1\rho_2+1)+1\big)},
\end{align}
where
\begin{align}\nonumber
&\gamma_{4,1}=6\rho_2^7+27\rho_2^6+53\rho_2^5+58\rho_2^4+36\rho_2^3+11\rho_2^2+\rho_2,~~\\\nonumber&
\gamma_{4,2}= 6\rho_2^7+48\rho_2^6+137\rho_2^5+193\rho_2^4+145\rho_2^3+55\rho_2^2+8\rho_2,\\\nonumber&
\gamma_{4,3}=2\rho_2^7+32\rho_2^6+143\rho_2^5+286\rho_2^4+287\rho_2^3+140\rho_2^2+26\rho_2,~~\\\nonumber&
\gamma_{4,4}=8\rho_2^6+67\rho_2^5+201\rho_2^4+281\rho_2^3+183\rho_2^2+43\rho_2,\\\nonumber&
\gamma_{4,5}=12\rho_2^5+67\rho_2^4+137\rho_2^3+123\rho_2^2+38\rho_2,~~~~~~~\\\nonumber&
\gamma_{4,6}=8\rho_2^4+31\rho_2^3+40\rho_2^2+14\rho_2,~~~~~~~~
\gamma_{4,7}=2\rho_2^3+5\rho_2^2+3\rho_2.
\end{align}
\begin{align}\label{VAoI5}
\bar{v}_{50}=\dfrac{\rho_2^6+3\rho_2^5+4\rho_2^4+3\rho_2^3+\rho_2^2+\sum_{k=1}^{7}\rho_1^k\gamma_{5,k}}{\mu(1+\rho)(1+\rho_2)(1+\rho_1)^2\big((\rho+1)^2-\rho_2\big)\big(\rho^2+\rho(2\rho_1\rho_2+1)+1\big)},
\end{align}
where
\begin{align}\nonumber
&\gamma_{5,1}=6\rho_2^6+22\rho_2^5+34\rho_2^4+28\rho_2^3+11\rho_2^2+\rho_2,~~~~~~\\\nonumber&
\gamma_{5,2}=7\rho_2^6+47\rho_2^5+103\rho_2^4+105\rho_2^3+49\rho_2^2+7\rho_2,\\\nonumber&
\gamma_{5,3}=3\rho_2^6+38\rho_2^5+133\rho_2^4+190\rho_2^3+115\rho_2^2+23\rho_2,\\\nonumber&
\gamma_{5,4}=12\rho_2^5+78\rho_2^4+170\rho_2^3+145\rho_2^2+40\rho_2,~~~~
\gamma_{5,5}=18\rho_2^4+73\rho_2^3+95\rho_2^2+37\rho_2,\\\nonumber&
\gamma_{5,6}=12\rho_2^3+29\rho_2^2+17\rho_2,~~~~
\gamma_{5,7}=3\rho_2^2+3\rho_2.
\end{align}

\bibliographystyle{IEEEtran}
\begin{spacing}{1.55}
\bibliography{RBibliography}
\end{spacing}
\end{document}